\documentclass[12pt, draftclsnofoot, onecolumn]{IEEEtranTCOM}
\renewcommand{\baselinestretch}{1.5}

\usepackage{indentfirst}
\usepackage{multirow}
\usepackage{graphicx}
\usepackage{color}
\usepackage{xcolor}
\usepackage{colortbl}
\definecolor{mypink}{rgb}{.99,.91,.95}
\usepackage{listings}
\usepackage{enumerate}
\usepackage{subfigure}
\usepackage{amsmath,amsthm,amscd,amssymb,latexsym,upref,stmaryrd}
\usepackage[noend]{algpseudocode}
\usepackage{algorithmicx,algorithm}
\usepackage{amsfonts,hyperref,epsfig}
\usepackage{CJK}
\usepackage{float}
\usepackage{cite}
\usepackage{subfigure}
\usepackage{lipsum}
\usepackage{multicol}
\usepackage{setspace}
\usepackage{bm}
\usepackage{stfloats}
\usepackage[aboveskip=-0pt,belowskip=-10pt]{caption}
\usepackage{geometry}
\usepackage{mathrsfs}
\usepackage{soul}
\usepackage{ulem}
\newtheorem{Theorem}{Theorem}
\newtheorem{Lem}{Lemma}%
\newtheorem{remark}{Remark}

\def\BibTeX{{\rm B\kern-.05em{\sc i\kern-.025em b}\kern-.08em T\kern-.1667em\lower.7ex\hbox{E}\kern-.125emX}}

\begin{document}
\title{Communication and Localization with Extremely Large Lens Antenna Array}

\author{Jie~Yang, Yong~Zeng,
	Shi~Jin, Chao-Kai~Wen, Pingping~Xu
	\thanks{Jie~Yang, Yong~Zeng, Shi~Jin, and Pingping~Xu are with the National Mobile Communications Research Laboratory, Southeast University, Nanjing, China (e-mail: \{yangjie;yong\_zeng;jinshi;xpp\}@seu.edu.cn). Chao-Kai~Wen is with the Institute of Communications Engineering, National Sun Yat-sen University, Kaohsiung, 804, Taiwan (e-mail: chaokai.wen@mail.nsysu.edu.tw). }
}

\maketitle	
\vspace{-1cm}
\begin{abstract}

Achieving high-rate communication with accurate localization and wireless environment sensing has emerged as an important trend of beyond-fifth and sixth generation cellular systems. Extension of the antenna array to an extremely large scale is a potential technology for achieving such goals. However, the super massive operating antennas significantly increases the computational complexity of the system. Motivated by the inherent advantages of lens antenna arrays in reducing system complexity, we consider communication and localization problems with an \uline{ex}tremely large \uline{lens} antenna array, which we call ``ExLens''. Since radiative near-field property emerges in the setting, we derive the closed-form array response of the lens antenna array with spherical wave, which includes the array response obtained on the basis of uniform plane wave as a special case. Our derivation result reveals a window effect for energy focusing property of ExLens, which indicates that ExLens has great potential in position sensing and multi-user communication.
We also propose an effective method for location and channel parameters estimation, which is able to achieve the localization performance close to the Cram\'{e}r-Rao lower bound.
Finally, we examine the multi-user communication performance of ExLens that serves coexisting near-field and far-field users. Numerical results demonstrate the effectiveness of the proposed channel estimation method and show that ExLens with a minimum mean square error receiver achieves significant spectral efficiency gains and complexity-and-cost reductions compared with a uniform linear array.

\end{abstract}


\section{Introduction}
\begin{spacing}{1.55} 
In comparison with previous generations, 
the fifth generation (5G) mobile network is a major breakthrough because of the introduce of massive multiple-input multiple-output (MIMO), millimeter-wave (mmWave), and ultra-dense network \cite{5G,5G1}.
However, realizing the full vision of supporting Internet of Everything services to connect billions of people and machines remains a challenge for 5G. 
Thus, research communities worldwide have implemented initiatives to conceive the next-generation (e.g., the sixth generation (6G)) mobile communication systems \cite{6G0,6G3,6G4,6G,6G1}.
The requirement of various applications, such as extended reality, autonomous systems, pervasive health monitoring, and brain computer interactions,
are driving the evolution of 6G towards a more intelligent and software reconfigurable functionality paradigm
that can provide ubiquitous communications and also sense, control, and even optimize wireless environments.

To fulfill the visions of 6G for high throughput, massive connectivity, ultra-reliability, and ultra-low latency, 
on the one hand,
mmWave and Tera-Hertz (THz) frequencies will be exploited further,
furthermore, multiple frequency bands (e.g., microwave/mmWave/THz frequencies) must be integrated to provide seamless connectivity \cite{thz};
on the other hand,
the antenna deployment will evolve towards larger apertures and greater numbers, furthermore, the extremely large aperture array has been proposed to boost spatial diversity further \cite{elaa,subarray1,hy}.
Moreover, intelligent reflecting surfaces or reconfigurable intelligent surfaces, artificial intelligence, and integrated terrestrial-aerial-satellite networks are regarded as promising technologies towards 6G\cite{wu,han,tang,uav,uav2}.
However, many open problems need to be solved to reap the full benefits of the aforementioned techniques.
In particular, when the antenna dimension continues to increase, the range of the radiative near-field of the antenna array expands, and the user equipment (UE) and significant scatterers are likely to be located in the near-field of the array.
Consequently, the prominent uniform plane wave assumption will no longer hold for extremely large antenna arrays \cite{spherical}.
Moreover, the use of thousands or more active antenna elements will generate prohibitive cost in terms of hardware implementation, energy consumption, and signal processing complexity \cite{perre}.

The radiation field of an antenna array is divided into the near-field region and the far-field region via the Rayleigh distance \cite{rayleigh0,rayleigh}, which is given as 
$
R = {2D^2}/{\lambda},
$
where $D$ is the maximum dimension of the antenna array, and $\lambda$ is the wavelength. 
When the distance between the UE (or scatterer) and the base station (BS) is smaller than the Rayleigh distance, the UE (or scatterer) is located in the near-field region, where the spherical wave-front over the antenna array is observed.
For example, a uniform linear array (ULA) of 1 meter (m) that operates at $30$ GHz corresponds to a Rayleigh distance of approximately $200\,m$ and nullifies the uniform plane wave-front model usually assumed in prior research on wireless communications. 
Few works have considered the near-field
property for modeling and analyzing massive MIMO channels 
by proposing the ULA response vector \cite{zz} and analyzing the channel estimation performance \cite{hy} under the spherical wave assumption.
The spherical wave-front is also proven to provide an underlying generic parametric model for estimating the position of UE and scatterers \cite{ULA2,yxf}.
Several works start to investigate the localization potential with large advanced antenna arrays to realize the vision of multi-purpose services for 6G (joint communication, control, localization, and sensing)
\cite{lis,henk,loclens,loclens1,loclens2,loclens3}
in addition to communication capabilities.  
Concentrated and distributed large antenna arrays are compared in \cite{lis} in terms of localization performance.
The theoretical uplink localization and synchronization performance is analyzed in \cite{henk} for ULA illuminated by spherical waves.
Parameter-based localization methods are developed in  \cite{loclens,loclens1,loclens2} for lens antenna arrays in the far-field, 
\cite{loclens3} considers direct localization by utilizing the near-field property, and provides a coarse localization accuracy.

An effective solution to significantly reduce the system complexity and implementation cost caused by the large number of antennas and UE is to partition the antenna array into a few disjoint subarrays \cite{subarray1,subarray2}.
In this work, we propose an alternative solution by using the energy focusing property of an \uline{ex}tremely large \uline{lens} antenna array denoted as ``ExLens'', which can fully utilize the aperture offered by the large antenna arrays. 
Recent studies have confirmed that the signal processing complexity and radio frequency (RF) chain cost could be significantly reduced without notable performance degradation for mmWave and massive MIMO systems by utilizing lens antenna arrays \cite{lens1,lens2,zy1,zy2,zy3}.
Electromagnetic (EM) lenses can provide variable phase shifting for EM rays at different points on the lens aperture to achieve angle-dependent energy focusing property.
Therefore, lens antenna arrays can transform the signal from the antenna space to the beamspace (the latter has lower dimensions) to reduce the RF chains significantly.
In \cite{zy2,zy3}, the array responses of lens antenna arrays have been derived in closed-form as a ``sinc" function of the angle of arrival (AOA)/angle of departure (AOD) of the impinging/departure signals. 
However, existing research on the lens antenna arrays are limited to the far-field assumption.
To the best of the authors' knowledge,
the array response of an ExLens for the general spherical wave-front has not been reported in prior works, let alone conducting a study on 
multi-user communication with an ExLens in the coexistence of  near-field and far-field UE.

In this study,
we explore the property of ExLens illuminated by spherical waves, including the capabilities of localization and multi-user communication,
on the basis of the inherent localization information carried by spherical waves and
the great potentials of lens antenna arrays in reducing system complexity.	
In summary,
we derive a closed-form array response of an ExLens, based on which we develop an effective method to obtain location parameters together with channel gains.
On the one hand, we can realize localization with the estimated location parameters.
On the other hand, we can design data transmission 
with the reconstructed channel.
Our main contributions are presented as follows:
\begin{itemize}
	\item \textbf{Array Response}:
	We first derive the closed-form expression for the array response of ExLens by considering the general spherical wave-front for two different EM lens designs,
	and then reveal that the obtained array response (derived based on the spherical wave assumption) includes the ``sinc''- type array response\cite{zy2} (derived based on the uniform plane wave assumption) as a special case.
    Next, we analyze differences of the energy focusing characteristics of ExLens illuminated by the spherical and plane wave-fronts.
	The window focusing property in the near-field of ExLens shows its great potential for position sensing and multi-user communication.
	The approximation error of the derived closed-form array response is verified ignorable.

	\item \textbf{Position Sensing}:
	We analyze the uplink localization ability of an ExLens equipped at the BS.
	We first study the theoretical localization performance from 
	a Fisher information perspective and confirm that the localization performance improves 	 
	as the aperture of the lens antenna array increases.
    By exploring the energy focusing window of ExLens, we propose an effective parameterized estimation method to obtain 
	location parameters together with channel gains. 
	Thus, localization can be performed by directly reusing the communication signals.
	Comprehensive simulations show that	the localization performance of the proposed method is close to the Cram\'{e}r-Rao lower bound (CRLB) and the channel can also be effectively reconstructed.

	\item \textbf{Multi-user Communication}:
	We investigate the multi-user communication performance of ExLens with coexisting near-field and far-field UE and scatterers.
	Power-based antenna selection is applied to ExLens to reduce the number of RF chains, together with the maximal ratio combining (MRC)- and minimum mean square error (MMSE)-based combining schemes to maximize the sum-rate.	
	The multi-user
	communication performance of the ExLens with perfect and estimated channel state information (CSIs) are compared.	
	Simulation results verify the effectiveness of the proposed channel estimation method and show that the proposed ExLens with an MMSE 
    receiver achieves significant spectral efficiency gains and complexity-and-cost reductions compared with the benchmark ULA schemes, when serving coexisting near-field and far-field UE. 
	
\end{itemize}

The rest of this paper is organized as follows: In Section \uppercase\expandafter{\romannumeral2}, we introduce an ExLens mmWave system model and derive the closed-form expression of ExLens array response. The property of ExLens array response is explained in Section \uppercase\expandafter{\romannumeral3}. In Section \uppercase\expandafter{\romannumeral4}, we explore the localization capbility of ExLens and propose an effective method to obtain location parameters together with channel gains. In Section \uppercase\expandafter{\romannumeral5}, we analyze the multi-user communication performance of ExLens. Our simulation results are presented in Section \uppercase\expandafter{\romannumeral6}. We conclude the paper in Section \uppercase\expandafter{\romannumeral7}.

{\bf Notations}---In this paper, upper- and lower-case bold letters denote matrices and vectors, respectively. For a matrix $\mathbf{A}$, $\mathbf{A}^{-1}$, $\mathbf{A}^{\text{T}}$, and $\mathbf{A}^{\text{H}}$ represent inverse, transpose, and Hermitian operators, respectively.
$\mbox{blkdiag}(\mathbf{A}_1,\ldots,\mathbf{A}_k)$ denotes a block-diagonal matrix constructed
by $\mathbf{A}_1,\ldots,\mathbf{A}_k$.
For a vector $\mathbf{a}$, the L$_2$-norm is signified by $\|\mathbf{a}\|$.
For a complex value $c$,
the module is represented by $|c|$ and the real part is denoted by $\mathcal{R}\{c\}$.
For a real number $a$, $\lfloor a \rfloor$ denotes the largest integer that is not greater than $a$.
sinc$(\cdot)$ is the ``sinc" function defined as sinc$(x)=\sin(\pi x)/(\pi x)$.
$\mathbb{E}\{\cdot\}$ indicates the statistical expectation.
\end{spacing}

\vspace{-0.3cm}
\section{System Model}\label{system}
We consider a BS equipped with an ExLens in the two-dimensional coordinate system (Fig. \ref{fig:model}).
The EM lens is placed on the y-axis with physical length $D_y$ and is centered at the origin.
The antenna elements are placed on the focal arc, which is defined as a semi-circle around the center of the EM lens with radius $F$. 
As the aperture of an antenna array further increases,
UE and significant scatterers are likely to be located in the near-field of the array, where the uniform plane wave-front assumption no longer holds. 
Therefore, we consider the more general spherical wave-front, which leads to more novel phase design of the EM lens and has greater energy focus on the lens antenna array compared with plane wave-front \cite{zy2}.

\begin{figure}
	\vspace{-0.3cm}
	\centering
	\hspace{-1.7cm}
	\subfigure[ExLens]
	{
		\begin{minipage}[t]{0.6\textwidth}
			\centering
			\includegraphics[scale=0.5,angle=0]{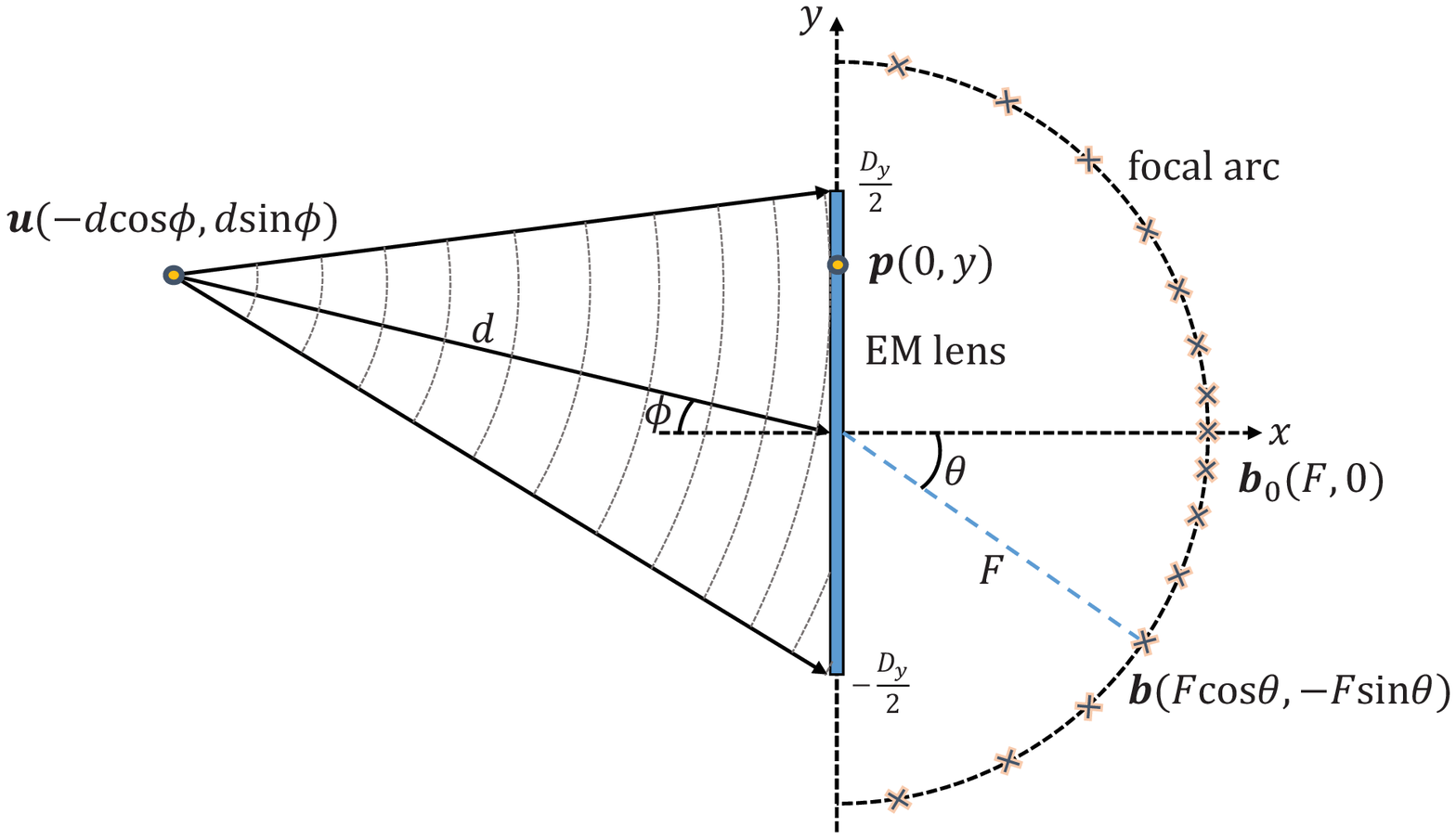}	\label{fig:model}
		\end{minipage}
	}
    \hspace{-1cm}	
	\subfigure[ULA]
	{
		\begin{minipage}[t]{0.36\textwidth}
			\centering
			\includegraphics[scale=0.5,angle=0]{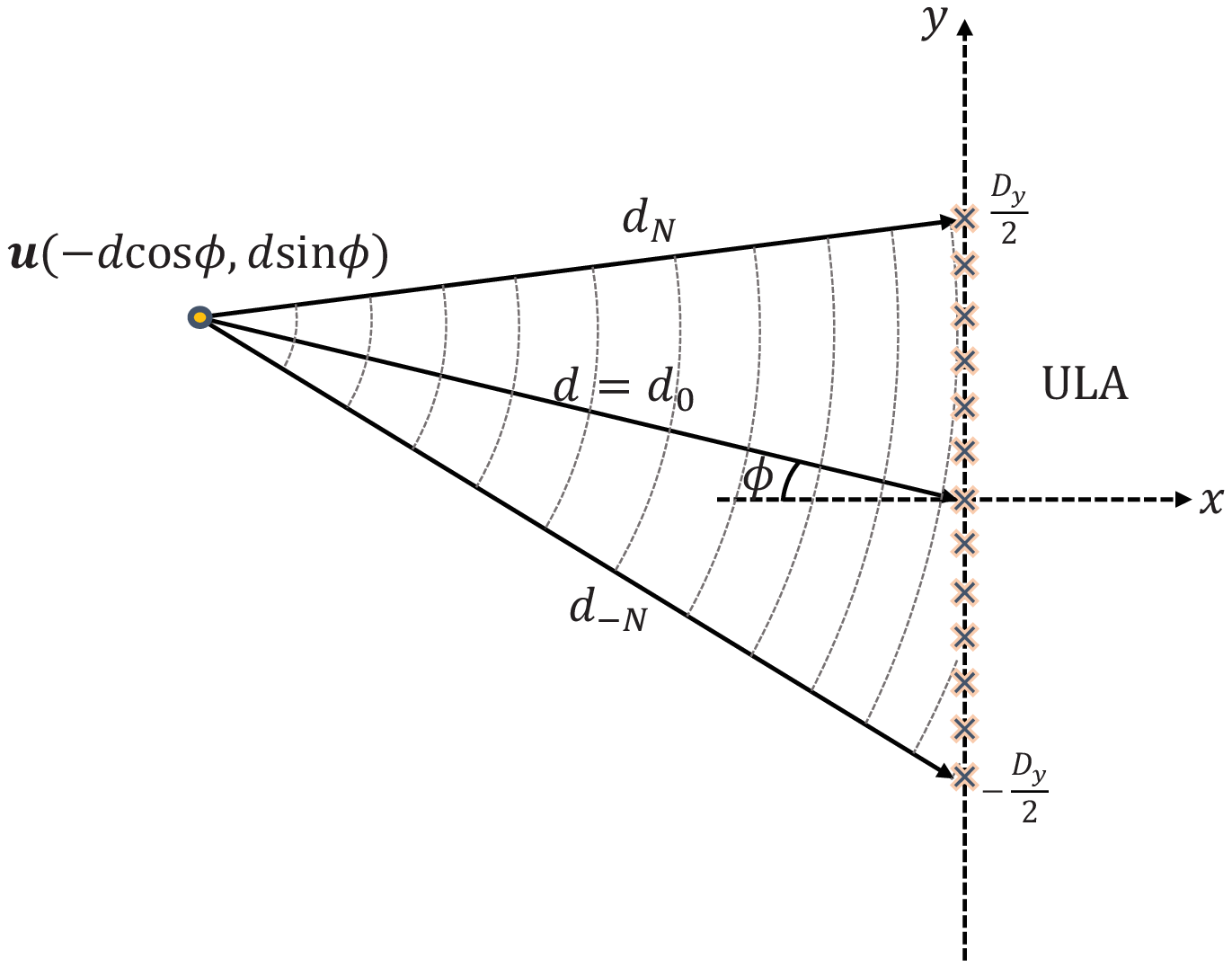}\label{fig:ULA}
		\end{minipage}
	}
	
	\captionsetup{font=footnotesize}
	\caption{Different antenna arrays illuminated by spherical wave-front.}
	\vspace{-0.7cm}
\end{figure}

We first investigate the receive array response by assuming that ExLens is illuminated by a spherical wave-front emitted from a UE located at $\mathbf{u}= [-d\cos\phi,d\sin\phi]$, where $d$ is the distance between the UE and the center of the EM lens, and $\phi \in (-{\pi}/{2},{\pi}/{2})$ is the angle of the UE relative to the x-axis (Fig. \ref{fig:model}). \footnote{
	We assume that the signal source is in front of the lens antenna array (i.e., it is located at the opposite side of the EM lens with the array elements). 
	This assumption practically holds because BSs apply sectorization technique. Each antenna array serves one sector in practice to cover the range of $60^\circ$ to $120^\circ$. Multiple lens antenna arrays can be combined to cover a range of $360^\circ$.
	}
For simplicity, we assume that the UE is equipped with an omni-directional antenna, and is regarded as a point source. 
The signal transmitted by the UE is assumed to be $1$, and the signal arrived at any point $\mathbf{p}=[0,y]$ on the EM lens aperture is given by \cite{zz,lis}
\vspace{-0.2cm}
\begin{equation}\label{s}
s(\mathbf{u},\mathbf{p}) = \eta(\mathbf{u},\mathbf{p})e^{-jk_0\parallel \mathbf{u} - \mathbf{p}\parallel},
\vspace{-0.4cm}
\end{equation}
where $k_0={2\pi}/{\lambda}$ is the wave number that 
corresponds to the signal wavelength $\lambda$, and $\eta(\mathbf{u},\mathbf{p})= {\lambda}/{(4\pi\!\parallel \!\mathbf{u} - \mathbf{p}\!\parallel})$ corresponds to the free space path loss from point $\mathbf{u}$ to point $\mathbf{p}$.
We define $\theta \in (-{\pi}/{2},{\pi}/{2})$, where $\theta$ is positive below the x-axis and negative above the x-axis (Fig. \ref{fig:model}).
The received signal $r(\theta,d,\phi)$ at any point $\mathbf{b}=[F\cos\theta,-F\sin\theta]$ at the focal arc \footnote{{The case that the center of the EM lens and the focal arc are not coinciding is left for future investigation.} } can be expressed as 
	\vspace{-0.2cm}
\begin{equation}\label{r}
r(\theta, d, \phi) = \int\limits_{ - {{{D_y}} \mathord{\left/
			{\vphantom {{{D_y}} 2}} \right.
			\kern-\nulldelimiterspace} 2}}^{{{{D_y}} \mathord{\left/
			{\vphantom {{{D_y}} 2}} \right.
			\kern-\nulldelimiterspace} 2}} {s(\mathbf{u},\mathbf{p})\kappa(\mathbf{p},\mathbf{b})e^{-j\varphi(\mathbf{p},\mathbf{b})}} dy = \int\limits_{ - {{{D_y}} \mathord{\left/
			{\vphantom {{{D_y}} 2}} \right.
			\kern-\nulldelimiterspace} 2}}^{{{{D_y}} \mathord{\left/
			{\vphantom {{{D_y}} 2}} \right.
			\kern-\nulldelimiterspace} 2}} { \eta(\mathbf{u},\mathbf{p})e^{-jk_0\parallel \mathbf{u} - \mathbf{p}\parallel}\kappa(\mathbf{p},\mathbf{b})e^{-j\varphi(\mathbf{p},\mathbf{b})}} dy.
		\vspace{-0.2cm}
\end{equation}
where
$\mathbf{p}$ is a function of $y$, $\mathbf{u}$ is a function of $(d, \phi)$, and $\mathbf{b}$ is a function of $\theta$; 
$\kappa(\mathbf{p},\mathbf{b})={\lambda}/{(4\pi\!\parallel \!\mathbf{p} - \mathbf{b}\!\parallel})$ accounts for the free-space path loss from point $\mathbf{p}$ on the EM lens to point $\mathbf{b}$ on the focal arc;
$\varphi(\mathbf{p},\mathbf{b})=\psi(\mathbf{p})+k_0||\mathbf{p}-\mathbf{b}||$, and $\psi(\mathbf{p})$ is the fixed phase shift determined by the EM lens design.
Therefore, $\varphi(\mathbf{p},\mathbf{b})$ is the total phase shift of the signal by the EM lens and the propagation delay between EM lens and focal arc. 
Eq. \eqref{r} follows the principle of linear superposition of signals.

\begin{figure}
	\centering
	\vspace{-0.5 cm}
	\subfigure[ Design 1: Incident plane wave-front perpendicular to the lens surface converges at the focal point, where ${\psi}(\mathbf{p}) = {\phi _0} - {k_0}\lVert \mathbf{p}-{\mathbf{b}_0}\rVert$.]
	{
		\begin{minipage}[t]{0.44\textwidth}
			\centering
			\includegraphics[scale=0.55,angle=0]{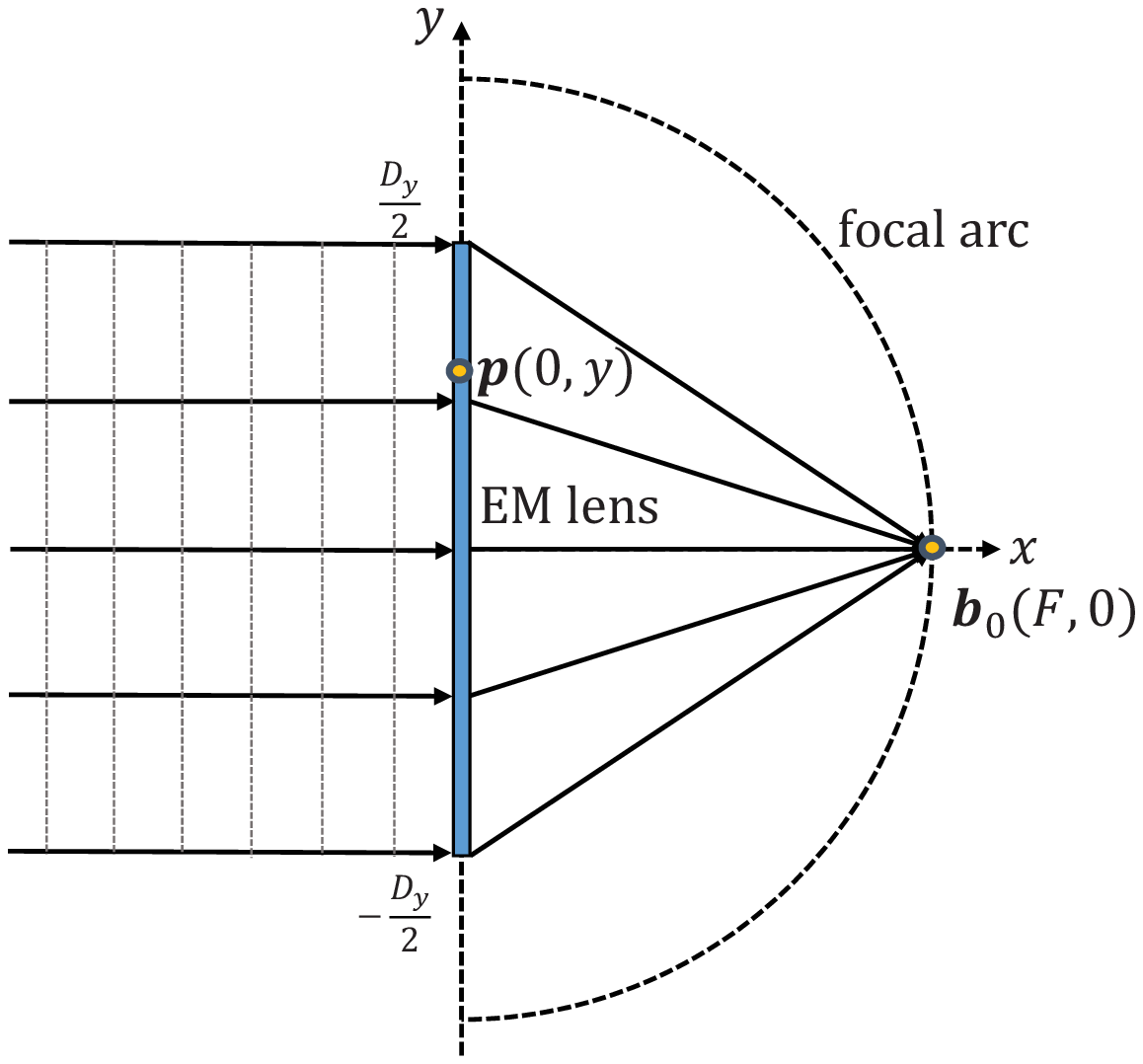}	\label{fig:d1}
		\end{minipage}
	}
	\hspace{0.25cm}	
	\subfigure[Design 2: Incident spherical wave-front from the left focal point converges at the right focal point, where ${\psi}(\mathbf{p}) = {\phi _0} - {k_0}(\lVert{\mathbf{c}_0}-\mathbf{p}\rVert + \lVert \mathbf{p}-\mathbf{b}_0\rVert)$.]
	{
		\begin{minipage}[t]{0.44\textwidth}
			\centering
			\includegraphics[scale=0.55,angle=0]{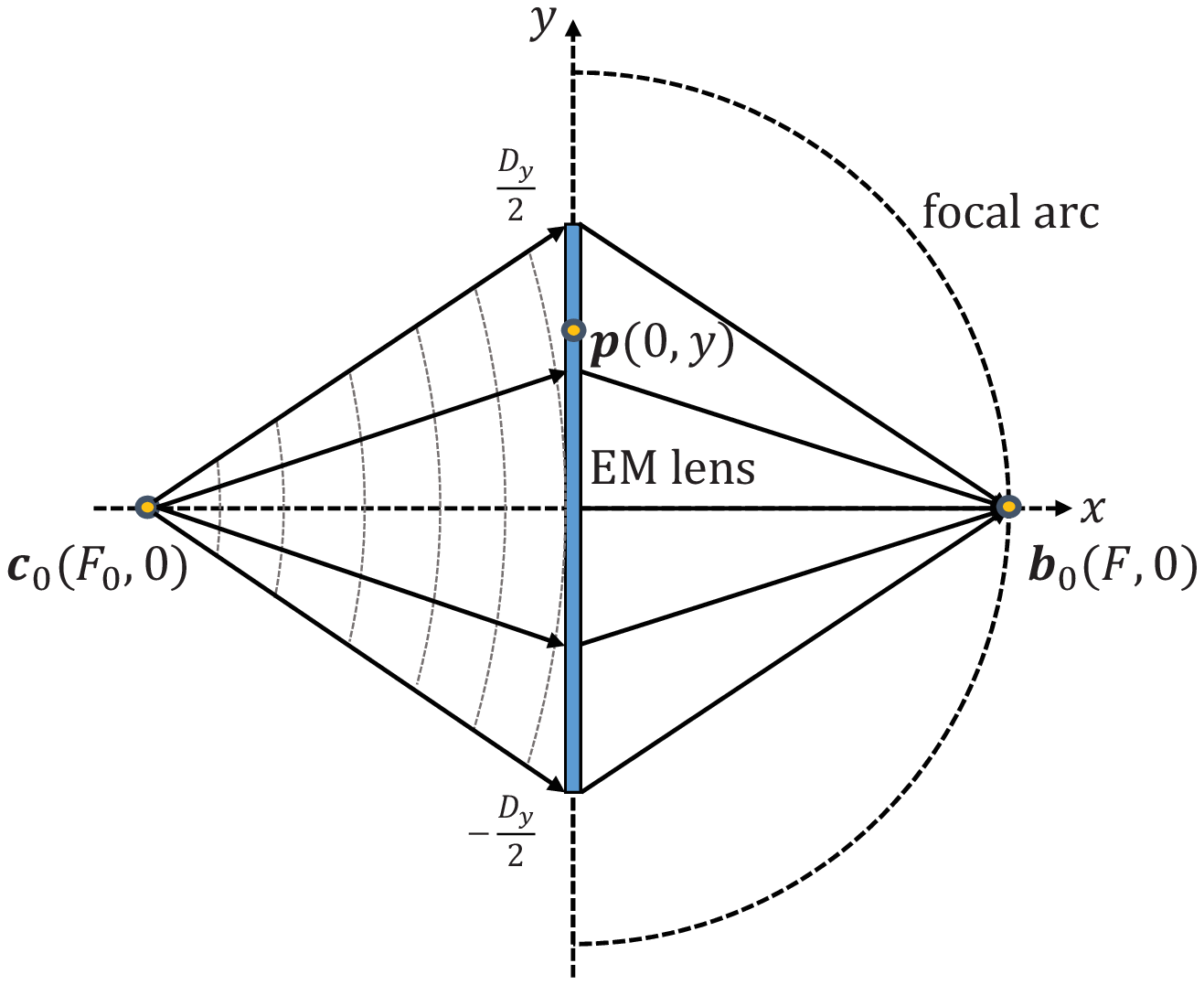}\label{fig:d2}
		\end{minipage}
	}
	
	\captionsetup{font=footnotesize}
	\caption{Two design approaches for the EM lens.}\label{lens}
	\vspace{-0.8cm}
\end{figure}

We first review the fundamental principle of operation for EM lenses:
the EM lenses are similar to optical lenses, which can alter the propagation directions of the EM rays to achieve energy focusing or beam collimation.
The function of EM lens can be effectively realized by appropriately designing $\psi(\mathbf{p})$ in $\varphi(\mathbf{p},\mathbf{b})$ in \eqref{r}.
In this study, we consider two different EM lens designs (Fig. \ref{lens}).
In Design 1, where incident plane wave-front perpendicular to the lens surface converges at the focal point $\mathbf{b}_0=[F,0]$ (Fig. \ref{fig:d1}), we have 
${\phi _0} = {\psi}(\mathbf{p}) + {k_0}\lVert \mathbf{p}-{\mathbf{b}_0}\rVert$,
where the constant $\phi_0$ is the arrived signal phase at the focal point $\mathbf{b}_0$.
Hence, we obtain
\vspace{-0.5cm}
\begin{equation}\label{pp}
{\psi}(\mathbf{p}) = {\phi _0} - {k_0}\lVert \mathbf{p}-{\mathbf{b}_0}\rVert.
\vspace{-0.4cm}
\end{equation}
The total phase shift for a signal from point $\mathbf{p}$ to point $\mathbf{b}$ is given by
\vspace{-0.25cm}
\begin{equation}\label{p1}
\begin{array}{l}
\varphi (\mathbf{p},\mathbf{b})= {\psi}(\mathbf{p}) + {k_0}\lVert \mathbf{p}-\mathbf{b}\rVert= {\phi _0} + {k_0}\left( \lVert \mathbf{p}-\mathbf{b}\rVert - \lVert \mathbf{p}-\mathbf{b}_0\rVert \right).
\end{array}
\vspace{-0.25cm}
\end{equation}
In Design 2, where the incident spherical wave-front from point $\mathbf{c}_0=[F_0,0]$ converges at the focal point $\mathbf{b}_0=[F,0]$ (Fig. \ref{fig:d2}), we have
${\phi _0} = {k_0}\lVert{\mathbf{c}_0}-\mathbf{p}\rVert + {\psi}(\mathbf{p}) + {k_0}\lVert \mathbf{p}-\mathbf{b}_0\rVert$.
Then, we obtain the following:
\vspace{-0.25cm}
\begin{equation}\label{pp2}
{\psi}(\mathbf{p}) = {\phi _0} - {k_0}(\lVert{\mathbf{c}_0}-\mathbf{p}\rVert + \lVert \mathbf{p}-\mathbf{b}_0\rVert).
\vspace{-0.5cm}
\end{equation}
We also obtain the total phase shift as follows:
\vspace{-0.5cm}
\begin{equation}\label{p3}
\begin{array}{l}
\varphi (\mathbf{p},\mathbf{b}) = {\psi}(\mathbf{p}) + {k_0}\lVert \mathbf{p}-\mathbf{b}\rVert = {\phi _0} + {k_0}(\lVert \mathbf{p}-\mathbf{b}\rVert - \lVert \mathbf{p}-{\mathbf{b}_0}\rVert - \lVert{\mathbf{c}_0}-\mathbf{p}\rVert).  
\end{array}
\vspace{-0.25cm}
\end{equation}
Design 1 can be regarded as a special case of Design 2 with $F_0 \rightarrow \infty $.
The EM lens is designed according to \eqref{pp} (Design 1) or \eqref{pp2} (Design 2) and works in the spherical wave-front scenarios (Fig. \ref{fig:model}).
Then, we define the response on point $\mathbf{b}=[F\cos\theta,-F\sin\theta]$ at the focal arc as 
\vspace{-0.25cm}
\begin{equation}\label{define}
a(\theta,d,\phi)={16\pi^2 F d}/({{\lambda^2 {e^{ - j{k_0 d}}} }}) \times r(\theta,d,\phi), 
\vspace{-0.25cm}
\end{equation}
the closed-form expression of which can be obtained in the following theorem.

\vspace{-0.3cm}
\begin{Theorem}\label{t1}
	When illuminated by a spherical wave-front, with the assumption $d,F\gg D_y$, the array response of ExLens on any point $\mathbf{b}=[F\cos\theta,-F\sin\theta]$ at the focal arc can be approximated as
	\vspace{-0.15cm}
	\begin{equation}\label{response}
	a(\theta,d,\phi)\approx\frac{{\sqrt \pi  }}{{{\rm{2}}\sqrt \alpha }}{e^{ - j\left( {\frac{{{{ {\pi^2 \beta^2} }}}}{{\alpha}} - \frac{5\pi }{4}} \right)}}\left( {\mathrm{erf}\left( {\frac{{\alpha{D_y} + 2\pi \beta}}{{2\sqrt \alpha }}{e^{j\frac{{3\pi }}{4}}}} \right) + \mathrm{erf}\left( {\frac{{\alpha{D_y} - 2\pi \beta}}{{2\sqrt \alpha }}{e^{j\frac{{3\pi }}{4}}}} \right)} \right),
	\vspace{-0.15cm}
	\end{equation}
	where
	\vspace{-0.15cm}
	\begin{equation}\label{erf}
	\mathrm{erf}\left( x \right) = \frac{2}{{\sqrt \pi  }}\int\limits_0^x {{e^{ - {t^2}}}dt},
	\vspace{-0.15cm}
	\end{equation}
	$\beta=(\sin \theta  - \sin \phi) /{\lambda }$ and $\alpha$ for the two different lens designs is given in Table \ref{af}.	
		\begin{table}[htbp]
		\centering
		\caption{Parameter $\alpha$ for different lens designs.}\label{af}
		\begin{tabular}{|c|c|c|}
			\hline
			\hline
			& Design 1 & Design 2 \\
			\hline
			$\alpha$ & $\dfrac{{\pi {{\sin }^{\rm{2}}}\theta }}{{\lambda F}} - \dfrac{{\pi {{\cos }^{\rm{2}}}\phi }}{{\lambda d}}$ & $\dfrac{{\pi {{\sin }^{\rm{2}}}\theta }}{{\lambda F}} - \dfrac{{\pi {{\cos }^{\rm{2}}}\phi }}{{\lambda d}} + \dfrac{\pi }{{\lambda {F_0}}}$ \\
			\hline	
		\end{tabular}
	\end{table}	
\end{Theorem}
\vspace{-1cm}
\begin{proof}
	Please refer to Appendix \ref{A}.
\end{proof}
According to Theorem \ref{t1}, the parameter $\alpha$ of Design 2 reduces to that of Design 1 when $F_0\to \infty$, as $\lim \limits_{F_0 \to \infty}{\pi }/(\lambda {F_0})= 0$ in Table \ref{af}.
This again shows that Design 1 can be regarded as a special case of Design 2 with $F_0 \rightarrow \infty $.
The energy focusing property of ExLens is determined by the item ${\mathrm{erf}\left( {\frac{{\alpha{D_y} + 2\pi \beta}}{{2\sqrt \alpha }}{e^{j\frac{{3\pi }}{4}}}} \right) + \mathrm{erf}\left( {\frac{{\alpha{D_y} - 2\pi \beta}}{{2\sqrt \alpha }}{e^{j\frac{{3\pi }}{4}}}} \right)}$ in \eqref{response} and is further analyzed in Section \ref{property}.
The parameter $\theta$ in Theorem \ref{t1} is a continuous value, whereas $\theta$ should be sampled for a particular antenna placement.
Here, we assume that $N_a = 2 \lfloor\tilde{D}_y\rfloor + 1$ antenna elements are placed on the focal arc of the EM lens \cite{zy2}, where $\tilde{D}_y={D_y}/{\lambda}$ denotes the electrical length of the EM lens.
For notational convenience, $N_a$ is assumed to be an odd number.
Let $\theta_n$ signify the angle of the $n$-th antenna element relative to the x-axis, where $n \in \{0,\pm 1, \ldots,\pm N\}$ and $N={(N_a-1)}/{2}$.
The deployment of the antenna elements obeys the rule
$\sin \theta_n = {n}/{N}$.
Therefore, the array response of the $n$-th antenna element located at point $\mathbf{b}_n=[F\cos\theta_n,-F\sin\theta_n]$ accroding to  \eqref{response} can be expressed as 
\vspace{-0.35cm}
\begin{equation}\label{sv}
a_n(d,\phi)\approx\frac{{\sqrt \pi  }}{{{\rm{2}}\sqrt \alpha }}e^{ - j\left( {\frac{{{{ {\pi^2 \beta^2} }}}}{{\alpha}} - \frac{5\pi }{4}} \right)}\left( {\mathrm{erf}\left( {\frac{{\alpha{D_y} + 2\pi \beta}}{{2\sqrt \alpha }}{e^{j\frac{{3\pi }}{4}}}} \right) + \mathrm{erf}\left( {\frac{{\alpha{D_y} - 2\pi \beta}}{{2\sqrt \alpha }}{e^{j\frac{{3\pi }}{4}}}} \right)} \right),
\vspace{-0.15cm}
\end{equation}
where $\sin \theta$ in $\alpha$ and $\beta$ is replaced by $\sin \theta_n = {n}/{N}$.
With the $n$-th element given in \eqref{sv}, the antenna array response vector $\mathbf{a}(d,\phi) \in \mathbb{C}^{N_a\times 1}$ can be obtained accordingly.

\section{Property of ExLens Array Response}\label{property}
In this section, we analyze the relationship and differences between the array responses of the lens antenna array illuminated by spherical wave-fronts (near-field scenarios) and plane wave-fronts (far-field scenarios).
Before entering the in-depth comparison, we review the array response of the lens antenna array illuminated by plane wave-fronts.
The array response of a lens antenna array
for an element located at the focal arc with angle $\theta$ and illuminated by a uniform plane wave with AOA $\phi$ is given by \cite{zy2}
\vspace{-0.5cm}
\begin{equation}\label{plane}
a(\theta, \phi) ={D_y}\mathrm{sinc}\left( {\tilde{D}_y\sin \theta  - \tilde{D}_y\sin \phi } \right).
\vspace{-0.45cm}
\end{equation}
In the far-field scenarios, the array response follows the ``{\rm sinc}" function as given in \eqref{plane}.
For any incident/departure signal from/to a particular direction $\phi$, only those antennas located near the focal point would receive/steer significant power. Notably, the focal point reflects the information of $\phi$, whereas the information of $d$ cannot be reflected from \eqref{plane}.
Furthermore, the angle resolution of the lens antenna array is determined by the width of the main lobe of the ``{\rm sinc}" function, which is $2/\tilde{D}_y$.
When $\tilde{D}_y$ increases, the main lobe becomes narrower such that other multi-paths can be resolved in the spatial domain.

\subsection{Generality Analysis}
We reveal the relationship between the array responses of the lens antenna array illuminated by the spherical wave-front in \eqref{response} and plane wave-front in \eqref{plane}.
The following lemma shows that \eqref{plane} is a special case of the derived ExLens array response \eqref{response}.
 
\vspace{-0.3cm}
\begin{Lem}\label{lem1}
	When $d$ and $F_0$ (for Design 2) increase to infinite and in which the spherical wave-front reduces to the plane wave-front,
	the array response given in \eqref{response} converges to \eqref{plane} as
	\vspace{-0.25cm}
	\begin{equation}\label{sv1}
	\hspace{-0.3cm}
	\begin{array}{ll}
	\lim \limits_{d, F_0 \to \infty}\! \frac{{\sqrt \pi  }}{{{\rm{2}}\sqrt \alpha }}e^{ - j\left( {\frac{{{{ {\pi^2 \beta^2} }}}}{{\alpha}} - \frac{5\pi }{4}} \right)}\!\!\left( {\mathrm{erf}\left( {\frac{{\alpha{D_y} + 2\pi \beta}}{{2\sqrt \alpha }}{e^{j\frac{{3\pi }}{4}}}} \right)\!\! + \mathrm{erf}\left( {\frac{{\alpha{D_y} - 2\pi \beta}}{{2\sqrt \alpha }}{e^{j\frac{{3\pi }}{4}}}} \right)} \right)
	\!=\!{D_y}\mathrm{sinc}\left( {\tilde{D}_y\sin \theta  - \tilde{D}_y\sin \phi } \right).
	\end{array}
	\vspace{-0.25cm}
	\end{equation}	
\end{Lem}
\vspace{-0.4cm}
\begin{proof}
Refer to Appendix \ref{B}.
\end{proof} 
	\vspace{-0.25cm}
\begin{remark}	
Lemma \ref{lem1} reveals that the derived array response of the ExLens illuminated by a spherical wave-front in \eqref{response} is a more general result compared to the result in \cite{zy2}, which means that the derived array response \eqref{response} is applicable to far-field (plane wave-front) and near-field (spherical wave-front) scenarios.
\end{remark}

\vspace{-0.7cm}
\subsection{Window Effect}
We first analyze differences of the \textbf{energy focusing characteristics} of the lens antenna array illuminated by the spherical wave-front in \eqref{response} and plane wave-front in \eqref{plane}.
Specifically, in the near-field scenarios,
the array response \eqref{response} has an evident window effect for the energy focusing property, which does not exist in the far-field scenarios.
To better understand the window effect, we split the array response \eqref{response} into three parts as follows:
\vspace{-0.2cm}
\begin{equation}\label{svi}
a(\theta ,d,\phi )=\underbrace{\frac{{\sqrt \pi  }}{{{\rm{2}}\sqrt \alpha }}}_{\left( {\rm{a}} \right)}\underbrace{e^{ - j\left( {\frac{{{{ {\pi^2 \beta^2} }}}}{{\alpha}} - \frac{5\pi }{4}} \right)}}_{\left( {\rm{b}} \right)}\underbrace{\left( {\mathrm{erf}\left( {\frac{{\alpha{D_y} + 2\pi \beta}}{{2\sqrt \alpha }}{e^{j\frac{{3\pi }}{4}}}} \right) + \mathrm{erf}\left( {\frac{{\alpha{D_y} - 2\pi \beta}}{{2\sqrt \alpha }}{e^{j\frac{{3\pi }}{4}}}} \right)} \right)}_{\left( {\rm{c}} \right)},
\vspace{-0.2cm}
\end{equation} 
where part $(a)$ is the amplitude, part $(b)$ is the phase, and part $(c)$ is the window effect for the energy focusing property in the near-field scenarios.
We denote a ``{\rm window}" function as
\begin{equation}\label{w1}	
{w(\theta ,d,\phi)}  \buildrel \Delta \over =  {\mathrm{erf}\left( {\frac{{\alpha{D_y} + 2\pi \beta}}{{2\sqrt \alpha }}{e^{j\frac{{3\pi }}{4}}}} \right) + \mathrm{erf}\left( {\frac{{\alpha{D_y} - 2\pi \beta}}{{2\sqrt \alpha }}{e^{j\frac{{3\pi }}{4}}}} \right)}.
\vspace{-0.2cm}
\end{equation}
\vspace{-1cm}
\begin{Lem}\label{lem2}
Let $v_1$ and $v_2$ denote the zero points of $\mathrm{erf}(\xi_1)$ and $\mathrm{erf}(\xi_2)$, respectively,	
where $\xi_1 = {\frac{{\alpha{D_y} + 2\pi \beta}}{{2\sqrt \alpha }}{e^{j\frac{{3\pi }}{4}}}}$ and $\xi_2 = {\frac{{\alpha{D_y} - 2\pi \beta}}{{2\sqrt \alpha }}{e^{j\frac{{3\pi }}{4}}}}$.	
We define the center and width of the energy focusing window as ${v_c} = ({v_1} + {v_2})/{2}$ and $\Delta v = |v_1 - v_2|$, respectively.
In lens Design 1, the center and width of the energy focusing window are given as follows:
\vspace{-0.25cm}
\begin{equation}\label{width1}
{v_c} = \sin\phi,\ \ \Delta v = \dfrac{{{D_y}{{\cos }^2}\phi }}{d}.
\vspace{-0.25cm}
\end{equation}
In lens Design 2, the center and width of the energy focusing window are obtained as follows:
\vspace{-0.25cm}
\begin{equation}\label{width2}
{v_c} = \sin\phi,\ \ \Delta v = {D_y}\left| {\dfrac{1}{{{F_0}}} - \dfrac{{{{\cos }^2}\phi }}{d}} \right|.
\vspace{-0.25cm}
\end{equation}
\end{Lem}
\begin{proof}
	Refer to Appendix \ref{C}.
\end{proof}
\vspace{-0.25cm}
\begin{remark}
	The \textbf{similarity} of the energy focusing property of the lens antenna array illuminated by the spherical wave-front in \eqref{response} and plane wave-front in \eqref{plane} is attributed
	to the following: the center of the focusing area is approximately equal to $\sin\phi$, where $\phi$ is the angle of the source point relative to the x-axis.	
The \textbf{differences} of the energy focusing property of the ``{\rm window}" function in \eqref{response} (near-field scenarios) and that of the ``{\rm sinc}" function \eqref{plane} (far-field scenarios) are as follows:
The width of the focusing window reflects the distance information $d$, according to \eqref{width1} and \eqref{width2}.
This feature implies that a single ExLens has the positioning capability in the spherical wave-front scenarios.
Therefore, the position of the UE can also be easily extracted from the information of the energy focusing windows according to the received communication signals.
By contrast, the ``{\rm sinc}" function has a maximum energy point and cannot reflect the information of $d$, according to \eqref{plane}.
\end{remark}
\vspace{-0.5cm}

Next, we explore the changes in the energy focusing properties from the far-field to the near-field, as aperture size of the lens antenna array increases.
For illustration, we take the lens Design 2 as an example and plot changes in the energy focusing property by increasing the effective aperture $\tilde{D}_y$ (Fig. \ref{fig:window2}).
Although $\mathrm{erf}(\xi_1)$ and $\mathrm{erf}(\xi_2)$ are complex values, their imaginary parts are close to zero.
Hence we draw the real parts of $\mathrm{erf}(\xi_1)$ and $\mathrm{erf}(\xi_2)$ on the right-hand side of Fig. \ref{fig:window2}.
When $\tilde{D}_y$ is small ($\tilde{D}_y=5$ and $\tilde{D}_y=10$), 
$d=50\,m$ is much larger than the Rayleigh distance $R = {2D_y^2}/{\lambda}$, the plane wave-front assumption holds,
and $|\mathrm{erf}(\xi_1)+\mathrm{erf}(\xi_2)|$ on the left-hand side of Fig. \ref{fig:window2} presents a representation of the ``sinc" function.
As long as the plane wave-front assumption holds, the ``sinc" function will become finer and sharper as $\tilde{D}_y$ increases.
Thus, the focusing property and the angle resolution of the lens antenna array improves, as described in \cite{zy2}.
When $\tilde{D}_y$ further increases to a sufficiently large value, the plane wave-front assumption no longer holds.
Furthermore, the ``sinc" function cannot reflect the energy focusing property accurately, say, when $\tilde{D}_y\geqslant70$. 
With the further increase in $\tilde{D}_y$ ($\tilde{D}_y=100$ and $\tilde{D}_y=200$), 
the window effect for the energy focusing property appears.  
The energy received in the focusing window is approximatley equal and the energy of the side lobes is extremely small,
as shown in the last two subfigures of Fig. \ref{fig:window2}.
The width of the energy focusing window increases with $\tilde{D}_y$,
thereby receiving more energy, but pronouncing energy diffusion effect.
The aforementioned phenomenon is caused by the relative distance of the lines of $\mathrm{erf}(\xi_1)$ and $\mathrm{erf}(\xi_2)$.
The relative distance of the lines of $\mathrm{erf}(\xi_1)$ and $\mathrm{erf}(\xi_2)$ increases with $\tilde{D}_y$, thereby resulting different line shapes of $|\mathrm{erf}(\xi_1)+\mathrm{erf}(\xi_2)|$ (Fig. \ref{fig:window2}).
In summary, the energy focusing property of the lens antenna array is described by the ``sinc" function in the far-field scenarios and the ``window" function in the near-field scenarios. 
	\vspace{-0.25cm}
\begin{remark}	
	The derived array response of the ExLens illuminated by a spherical wave-front in \eqref{response} can describe the transition between the far-field and the near-field scenarios.
	In the special case of far-field (i.e., $d, F_0 \to \infty$), the angle resolution of the lens antenna array is determined by the width of the main lobe of the ``{\rm sinc}" function, which is $2/\tilde{D}_y$.
	By contrast, the width $\Delta v$ determines the angle resolution of the ExLens in the near-field.		 
	The angle resolution makes ExLens illuminated by spherical waves also suitable for multi-user communication. 
\end{remark}
\vspace{-0.25cm}
\begin{figure}[t]
	\vspace{-0.5cm}
	\centering
	\hspace{-0.31cm}
	\includegraphics[scale=0.6,angle=0]{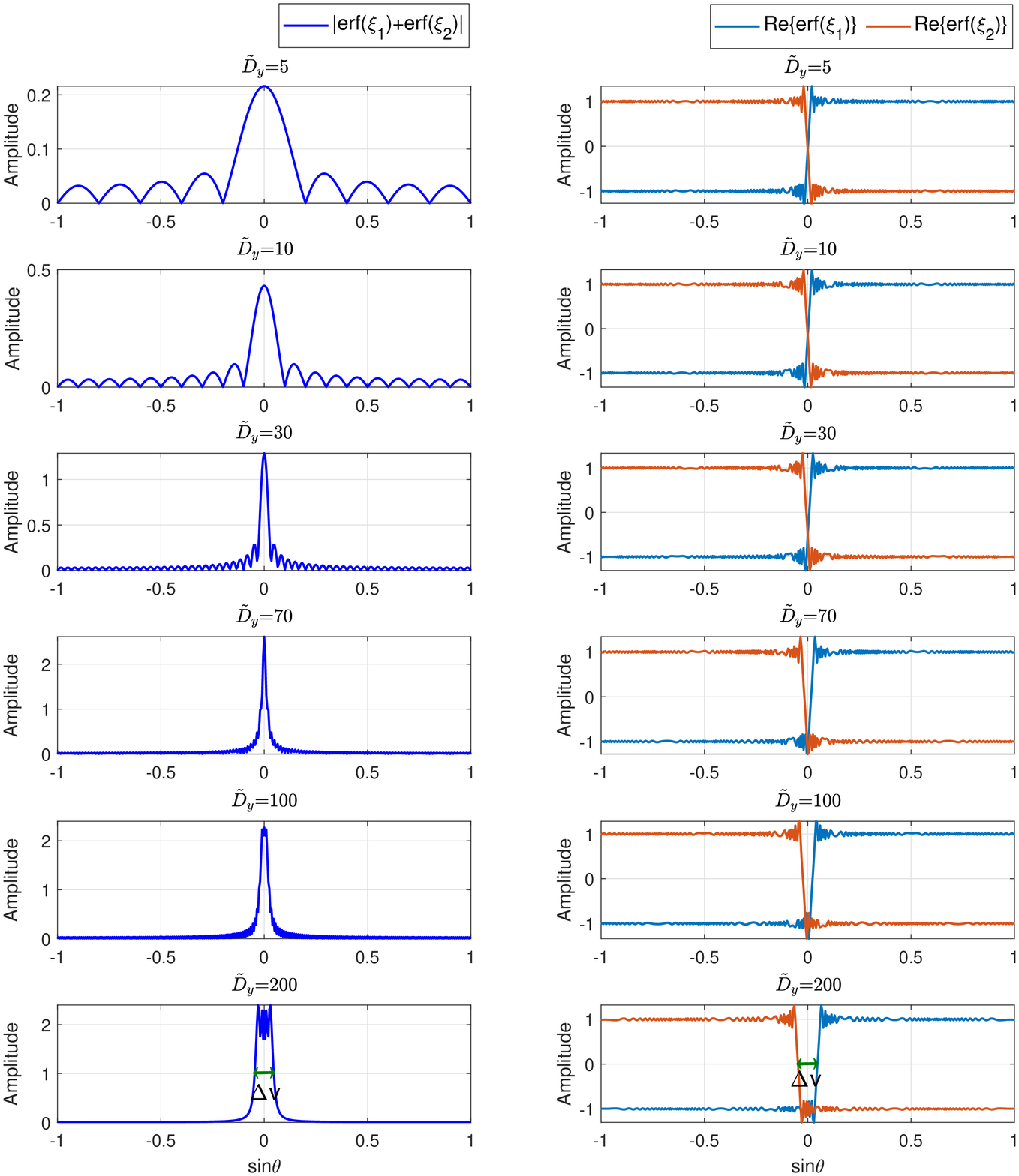}	
	\captionsetup{font=footnotesize}
	\caption{Changes in the energy focusing property of the lens antenna array as the effective aperture $\tilde{D}_y$ increases, where $\phi=0$ rad, $d=50\,m$, $F_0=15\,m$, $F=5\,m$, and $\lambda= 0.01\,m$.}
	\label{fig:window2}	
	\vspace{-0.8cm}
\end{figure}

	\vspace{-0.5cm}
\subsection{Approximation Tightness and Antenna Deployment}
\begin{figure}[t]
	\vspace{-0.1cm}
	\centering
    \includegraphics[scale=0.59,angle=0]{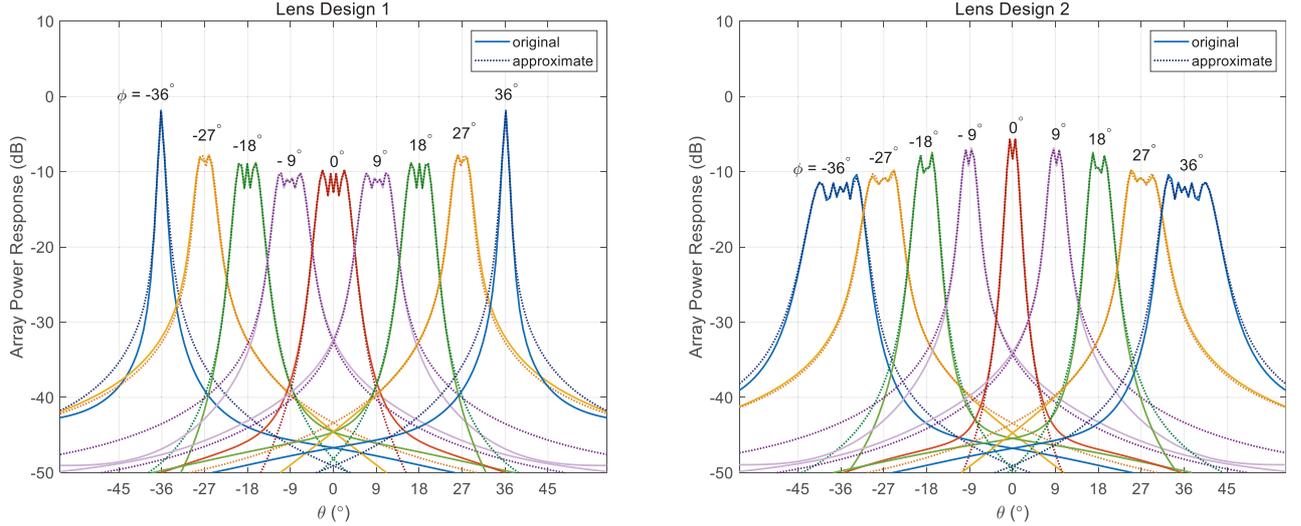}
    \captionsetup{font=footnotesize}
	\caption{Differences of the energy focusing properties of the ExLens with two different lens designs, where $\tilde{D}_y=100$, $d=7\,m$, $F_0=F=5\,m$, and $\lambda= 0.01\,m$.}\label{design}
	\vspace{-0.7cm}
\end{figure}
	The approximation error between the derived closed-form array response and the original one in integral form is compared in Fig. \ref{design}.
	The solid line denotes the original array response in the integral-form.
	The dotted line represents the approximated closed-form array response given in \eqref{response}.
	The solid line matches well with the dashed line when the array power response is above $-20$ dB, thereby indicating that the approximation error is small.
	Specifically, the approximation error can be safely ignored when $\phi \in [-36^\circ,36^\circ]$.
	When the array power response is below $-20$ dB, the approximate error slightly increases with $|\phi|$.
Moreover, the two different lens designs also have dissimilar energy focusing characteristics (Fig. \ref{design}).
For Design 1, the width of the energy focusing window $\Delta v = {{{D_y}{{\cos }^2}\phi }}/{d}$ is smaller with larger $|\phi|$, and this characteristic means that the energy focusing phenomenon is evident when $|\phi|$ is large.
However, for Design 2, we have $\Delta v = {D_y}\left| {{1}/{{{F_0}}} - {{{{\cos }^2}\phi }}/{d}} \right|$ and $d>{F_0}\cos^2\phi$, the energy focusing phenomenon is evident when $|\phi|$ is small.
This outcome is expected given that the ExLens of Design 1 has good energy focusing property for the uniform plane incident wave.
When we apply this design into spherical wave-front scenarios, the incident spherical wave-front becomes closer to the plane wave-front
as $|\phi|$ increases; hence, the energy focusing performance improves.
However, the ExLens of Design 2 has good energy focusing property when the source point is near the focal point  $\mathbf{c}_0=[-F_0,0]$. When $|\phi|$ becomes larger, the source point is farther from the focal point $\mathbf{c}_0$; thus, the energy focusing performance deteriorates.
As mentioned before, when $F_0\to \infty$, Design 2 becomes to Design 1.
Accordingly, the width of the energy focusing window of Design 2 equals to that of Design 1, i.e., $\lim\limits_{F_0 \to \infty}{D_y}\left| {{1}/{{{F_0}}} - {{{{\cos }^2}\phi }}/{d}} \right|={{{D_y}{{\cos }^2}\phi }}/{d}$. 
To be more general, we use the lens Design 2 for analysis in the following sections.
We adopt the antenna elements deployment $\sin \theta_n = {n}/{N}$, 
as mentioned in Section \ref{system}.
The right hand-side of Fig. \ref{design} shows that the energy focusing window is narrower in the center and wider on the edges.
Therefore, placing denser antenna elements in the center of the array can prevent the non-detection of strong signals.
For the far-field scenario, the antenna array response reduces to the ``sinc" function \cite{zy2}, and this kind of antenna deployment is applicable.

From the analysis in this section, we get such insight that the window effect for the energy focusing property makes ExLens illuminated by spherical waves suitable for single-station localization and multi-user communication,
which are analyzed in-depth in the following sections.

\section{Position Sensing}\label{ps}
In this section, we explore the localization capbility of ExLens.
The signal that arrived at different points of the lens aperture has the same incident angle under the assumption of plane wave-front.
However,
when the UE is located in the near-field of ExLens, the signal with a spherical wave-front arrived at different points of the lens aperture has different incident angles.
Thus, relative to that of the plane wave-front, the received signal with spherical wave-front contains more abundant angular information that changes continuously from one edge of the lens aperture to another.
According to the traditional multi-point localization \cite{localization}, more angular measurements can ensure more accurate localization.
We can infer that a single ExLens has the localization capbility with the abundant angular information from the spherical wave-front.
Thus, in this section, we analyze the theoretical localization capbility of ExLens and then propose a parameterized estimation method to obtain the location parameters.

For ease of expression, we take one UE with single antenna for illustration in this section.
The system model can be easily extended to solve the case with multiple UEs as long
as the pilot signals for different UEs are orthogonal in time.
We consider the narrow band mmWave multi-path channel model. 
Thus, the uplink channel is given by 
\vspace{-0.25cm}
\begin{equation}\label{ha}
\textbf{h} = \sum\limits_{l=1}^{L} g_l \textbf{a}(d_l,\phi_l),
\vspace{-0.25cm}
\end{equation} 
where $g_l$ is the complex gain of the $l$-th path, $\textbf{a}(\cdot)\in \mathbb{C}^{N_a\times 1}$ is the array response vector with elements defined in \eqref{sv},
$l = 1$ represents the LOS path, $(d_1,\phi_1)$ is a pair of position parameters of the UE,
$l > 1$ represents the NLOS path, and $(d_l,\phi_l)$ is a pair of position parameters of the $l$-th scatterer.
We only consider the last-jump scatterers.
If all one pilots are used,
then, the received signal at the ExLens antenna array is modelled as
\vspace{-0.25cm}
\begin{equation}\label{re}
\textbf{r} = \textbf{h} + \textbf{n} = \sum\limits_{l=1}^{L} g_l \textbf{a}(d_l,\phi_l) + \textbf{n},
\vspace{-0.25cm}
\end{equation} 
where $\textbf{n}\in \mathbb{C}^{N_a\times 1}$ represents the circularly symmetric complex Gaussian noise with zero-mean and covariance matrix $\sigma^2\textbf{I}$.
Here, we define the receive signal-to-noise ratio (SNR) as SNR $ = 10\lg({\textbf{h}^{\text{H}}\textbf{h}}/{(N_a \sigma^2)})$.
{Let $\bm{\eta}_l = [g_l, d_l,\phi_l]$, $\bm{\eta} = [\bm{\eta}_1, \ldots, \bm{\eta}_L]$,
and $\textbf{h}(\bm{\eta})=\sum\limits_{l=1}^{L} g_l \textbf{a}(d_l,\phi_l)$.}
In the localization, we aim at determining $\bm{\eta}$ based on the received signal $\textbf{r}$ in \eqref{re}.

\subsection{Theoretical Localization Analysis}
{According to \cite{KAY}, the $3L\times 3L$ positive definite Fisher information matrix (FIM) of $\bm{\eta}$ is given by
\begin{equation}\label{fim1}
{\bf{F}}\left( \bm{\eta}  \right) = \left[
\begin{matrix}
{\bf{F}}_{11}\left( \bm{\eta}  \right) & \ldots & {\bf{F}}_{1L}\left( \bm{\eta}  \right) \\
\vdots &\ddots &\vdots \\
{\bf{F}}_{L1}\left( \bm{\eta}  \right) & \ldots & {\bf{F}}_{LL}\left( \bm{\eta}  \right) 
\end{matrix}
\right] ,
\end{equation}
where the $3\times 3$ matrix ${\bf{F}}_{ll'}\left( {\bm{\eta} } \right)$ is defined by
\vspace{-0.25cm}
\begin{equation}\label{fim2}
{\bf{F}}_{ll'}\left( {\bm{\eta} } \right) =\frac{2}{\sigma^2} \mathcal{R}\left\{
\dfrac{\partial \textbf{h}^{\text{H}}(\bm{\eta}) }{\partial \bm{\eta}_l}  \dfrac{\partial \textbf{h}(\bm{\eta}) }{\partial \bm{\eta}_{l'}} 
\right\}.
\vspace{-0.25cm}
\end{equation}
The information inequality for the covariance
matrix of any unbiased estimate $\hat{\bm{\eta}}$ reads \cite{KAY}
\vspace{-0.3cm}
\begin{equation}\label{CRLB}
\mathbb{E}\{(\hat{\bm{\eta}} - \bm{\eta})^\text{H}(\hat{\bm{\eta}} - \bm{\eta})  \} \geq {\bf{F}^{-1}}\left( \mathbf{\bm{\eta}}  \right).
\vspace{-0.3cm}
\end{equation}
The ${\bf{F}}\left( \bm{\eta}  \right)$ is represented in the polar coordinates.
With $x_l=-d_l\cos\phi_l$ and $y_l=d_l\sin\phi_l$,
the position of the UE or scatterer in Cartesian coordinates is given as $\mathbf{u}_l=[x_l,y_l]$, for $l = 1, \ldots, L$.
Let $\mathbf{u} = [\mathbf{u}_1,\ldots, \mathbf{u}_L]$, $\tilde{\mathbf{u}}_l=[g_l,x_l,y_l]$,  and $\tilde{\mathbf{u}} = [\tilde{\mathbf{u}}_1,\ldots, \tilde{\mathbf{u}}_L]$,
thus,
transformation to the position domain is achieved as follows:
the FIM of $\tilde{\mathbf{u}}$  is given by
${\bf{F}}\left(\tilde{\mathbf{u}}  \right) = {\bf{T}}^{\text{T}}{\bf{F}}\left( {\bm{\eta} } \right){\bf{T}}$, 
where ${\bf{T}}= {\rm blkdiag}\{\mathbf{T}_1,\ldots, \mathbf{T}_L\}$, and ${\bf{T}}_l= [1,0,0; 0,x_l/d_l, y_l/d_l; 0,y_l/d_l^2,-x_l/d_l^2]$ is the Jacobian matrix used to describe the coordinate system transformation, in which the ``;" operator separates rows in a matrix.
Then, we define the position error bound (PEB) from ${\bf{F}}\left( \tilde{\mathbf{u}}  \right)$ as
\vspace{-0.35cm}
\begin{equation}\label{peb}
{\text{PEB}}(\mathbf{u}) = {\sqrt{{\rm trace}([{\bf{F}^{-1}}\left( \tilde{\mathbf{u}}  \right)]_{(\{2:3,5:6,\ldots,3L-1:3L\},\{2:3,5:6,\ldots,3L-1:3L\})})}}.
\vspace{-0.3cm}
\end{equation}
The root mean-square estimation error (RMSE) of an unbiased estimate of ${\mathbf{u}}$ is lower-bounded by ${\text{PEB}}(\mathbf{u})$.
We need ${{\partial \mathbf{a}(d_l,\phi_l)}}/{{\partial d_l}}$ and ${{\partial \mathbf{a}(d_l,\phi_l)}}/{{\partial \phi_l}}$, which is derived in Appendix \ref{FIM}, to calculate the PEB given in \eqref{peb}.}

Fig. \ref{PEB} shows PEBs as functions of $d$, $\phi$, $D_y$, and $F_0$ ($F$ shows a similar property as $F_0$).
Given that some approximations are made to derive the closed-form
array response \eqref{sv}, we denote the obtained PEB with the approximated closed-form array response as APEB. 
We also calculate the PEB with the original received signal in the integral form denoted as OPEB.
We can evaluate the accuracy of the approximation by comparing APEB with OPEB.
Given the minimal approximation error shown in Fig. \ref{PEB} (only when $|\phi|>1.2$ rad or ${D}_y > 3\,m$, the APEB slightly deviates from the OPEB),
the theoretical localization analysis based on \eqref{sv} is accurate.

\begin{remark}	
	\vspace{-0.2cm}	
	The localization performance of ExLens degrades with the increase in $d$ and $|\phi|$. 
    This finding is expected given that the UE transits from near-field to far-field because $d$ increases,
    thereby demonstrating that a single BS loses its localization capability.	 
	The localization performance improves with the increase in ${D}_y$.
	By contrast, the value of $F_0$ has a little effect on PEBs.
	Given that the increase in ${D}_y$ can bring rich angle measurements, which is similar to adding additional BSs in the multi-point localization.	
	Under the given configuration in Fig. \ref{PEB}, an ExLens with the electrical aperture $\tilde{D}_y = 100$ and SNR $= 20$ dB can theoretically provide around centimeter-level localization accuracy for a UE located with  $d < 50\,m$ and $\phi < 1$ rad to the BS.
	\vspace{-0.6cm}	
\end{remark}

\begin{figure}[t]
	\vspace{-0.5cm}
	\centering
	\includegraphics[scale=0.47,angle=0]{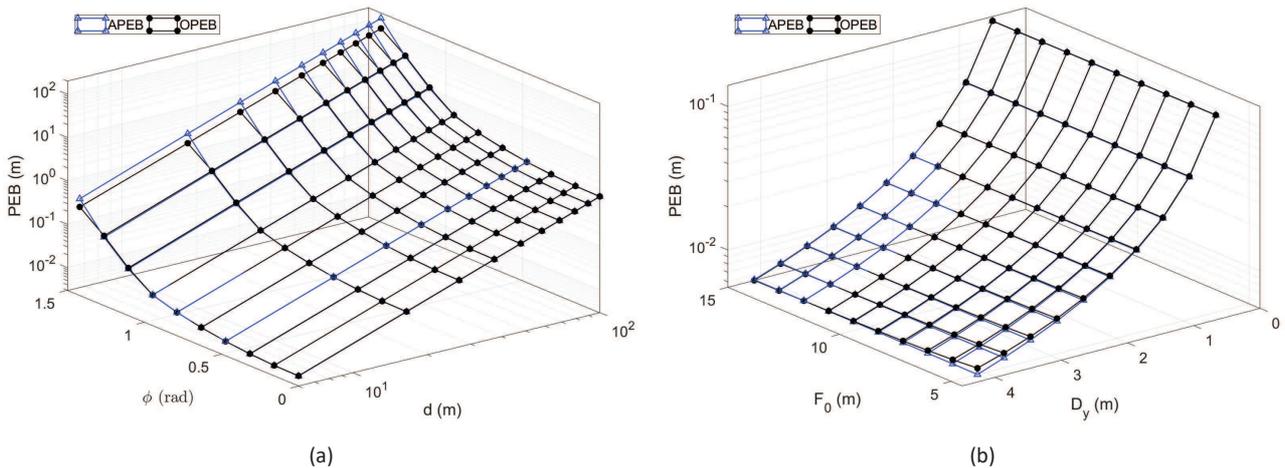}			\captionsetup{font=footnotesize}
	\caption{
		(a) $\mbox{PEB}$ as a function of $d$ and $\phi$ with $D_y=1\,m$, $F_0=F=5\,m$, $\lambda = 0.01 \,m$, $N_a = 201$, $L=1$, and SNR $=20$ dB.
		(b) $\mbox{PEB}$ as a function of ${D}_y$ and $F_0$ with $d=18\,m$, $\phi=0$ rad, $F=5\,m$, $\lambda = 0.01\,m$, $N_a = 201$, $L=1$, and SNR $=20$ dB.		
	}\label{PEB}
	\vspace{-0.5cm}
\end{figure}

\subsection{Location Parameter Estimation Method}\label{Loc}
In this subsection, we propose a location parameters estimation method to determine $(d_l,\phi_l)$,
and the gain $g_l$ for $l = 1, \ldots, L$.
The maximum likelihood (ML) estimator is given by
\vspace{-0.25cm}
\begin{equation}\label{ml1}
(\bm{\hat{d}},\bm{\hat{\phi}},\bm{\hat{g}})=\mathop{\arg\min}\limits_{ \bm{d}\in \mathbb{R}^L, \bm{\phi}\in(-\frac{\pi}{2}, \frac{\pi}{2})^L,\bm{g} \in \mathbb{C}^L} \left\Arrowvert \textbf{r} - \sum\limits_{l=1}^{L} g_l\textbf{a}(d_{l},\phi_{l}) \right\Arrowvert^2,
\vspace{-0.25cm}
\end{equation}
where ${\bm d} = [d_1,\ldots,d_L ]$, ${\bm \phi} = [\phi_1,\ldots,\phi_L]$, and ${\bm g} = [g_1,\ldots,g_L]$. 
\footnote{
	We assume that only $M_{RF}$ RF chains are available in the ExLens system, where $M_{RF}<N_a$.	
	Thus, the low-complexity power-based antenna selection method is applied.
	We let $\textbf{r} \in \mathbb{C}^{N_a\times 1}$ denote the received signal after the antenna selection, which has $M_{RF}$ non-zero elements.}
The brute-force search for the optimal estimate of $(\bm{d},\bm{\phi},\bm{g})$ in the whole continuous domain ($\bm{d}\in \mathbb{R}^L$, $\bm{\phi}\in(-{\pi}/{2},{\pi}/{2})^L$, and  $\bm{g} \in \mathbb{C}^L$) is infeasible.
Hence, we propose an effective localization method, which contains three stages: 
(1) Initialization stage, where we propose a window-based coarse localization algorithm to determine the grid search region.
(2) detection stage, in which we find a relatively accurate estimate of $(d_l,\phi_l)$, for $l = 1, \ldots, L$, from discrete grids by discrete OMP (DOMP) algorithm.
(3) refinement stage, where we iteratively refine the location parameters $(d_l,\phi_l)$ and gain $g_l$ for $l = 1, \ldots, L$ by Newton algorithm \cite{n1,n2}.

\subsubsection{Initialization stage}
We utilize the window effect for energy focusing property of ExLens to narrow down the search region. 
Lemma \ref{lem2} is developed for single-path scenarios.
For multi-path scenarios,
we let  $v_{1l}$ and $v_{2l}$ denote 
the window edges for the $l$-th path, which are affected by the position parameters $(d_l,\phi_l)$, where $l = 1, \ldots, L$.
In Appendix \ref{C}, we derive the relationships between the window edges and the location parameters.
After parameters $v_{1l}$ and $v_{2l}$ are measured, we can obtain a coarse estimation of $(d_l,\phi_l)$ for $l = 1, \ldots, L$.
We apply power detection to the received signal by each antenna, and for a given threshold, we can obtain $\hat{v}_{1l}$ and $\hat{v}_{2l}$ for $l = 1, \ldots, L$.
According to \eqref{w3} and \eqref{w5}, we have the following expression for the $l$-th path:
\vspace{-0.25cm}
\begin{equation}\label{wl1}
\mathbf{g}_l = \mathbf{G}\mathbf{q}_l+\mathbf{e}_l,
\vspace{-0.5cm}
\end{equation}
where
\vspace{-0.3cm}
\begin{equation}\label{wl2}
\mathbf{g}_l=\!\left(\!\left(\!\hat{v}_{1l}\!+\!\frac{F}{D_y}\right)^2\!\!\!-\!\left(\!\frac{F}{D_y}\!\right)^2\!\!\!+\!\frac{F}{F_0},\!\left(\!\hat{v}_{2l}\!-\!\frac{F}{D_y}\!\right)^2\!\!\!-\!\left(\frac{F}{D_y}\right)^2\!\!\!+\!\frac{F}{F_0}\right)^{\text{T}}\!\!\!,\ \ \mathbf{G}=\begin{pmatrix} \frac{2F}{D_y}& F\\
\frac{-2F}{D_y}& F
\end{pmatrix},\ \ \mathbf{q}_l=\begin{pmatrix}\sin\phi_l\\ \frac{\cos^2\phi_l}{d_l}\end{pmatrix},
\vspace{-0.1cm}
\end{equation}
and $\mathbf{e}_l$ is the noise vector caused by measurement error.
Then, the least squares estimator is given by 
\vspace{-0.3cm}
\begin{equation}\label{wl4}
\mathbf{\hat{q}}_l = (\mathbf{G}^{\text{H}}\mathbf{G})^{-1}\mathbf{G}^{\text{H}}\mathbf{g}_l,
\vspace{-0.3cm}
\end{equation}
with $\mathbf{\hat{q}}_l = [\hat{q}_{1l},\hat{q}_{2l}]$.
Thus, the position parameters $(d_l,\phi_l)$ can be recovered by
\vspace{-0.3cm}
\begin{equation}\label{wl5}
\hat{\phi}_l = \arcsin \hat{q}_{1l},\ \ \hat{d}_l = ({1-\hat{q}_{1l}^2})/{\hat{q}_{2l}}.
\vspace{-0.3cm}
\end{equation}
We denote the sets $\mathbb{S}_d = \cup \mathbb{S}_d^l$ and $\mathbb{S}_{\phi} =\cup \mathbb{S}_{\phi}^l$, for $l = 1, \ldots, L$, where 
$\mathbb{S}_d^l = \{ \hat{d}_l - \Delta d \leq d \leq \hat{d}_l+\Delta d \}$ and $\mathbb{S}_{\phi}^l = \{ \hat{\phi}_l - \Delta \phi \leq \phi \leq \hat{\phi}_l+\Delta \phi \}$. 
We generate finite discrete sets by taking $N_d$ and $N_\phi$ grids on the obtained sets $\mathbb{S}_d$ and $\mathbb{S}_{\phi}$ as $\mathbb{\bar{S}}_d$ and $\mathbb{\bar{S}}_{\phi}$, respectively.
The total search region is initialized as $\mathbb{\bar{S}}_d$ and $\mathbb{\bar{S}}_{\phi}$.

\subsubsection{Detection stage}
We apply DOMP algorithm to detect $d_l$ and $\phi_l$ from the discrete sets $\mathbb{\bar{S}}_d$ and $\mathbb{\bar{S}}_{\phi}$, respectively, for $l = 1, \ldots, L$.
We take the detection of  the $l'$-th path as an example for illustration.
Let $(\hat{d}_l,\hat{\phi}_l,\hat{g}_l)$, for $l = 1, \ldots, l'-1$, denote the estimates of the first $l'-1$ paths.
Then, the residual measurement is given by
\vspace{-0.25cm}
\begin{equation}\label{ml3}
\textbf{r}_r = \textbf{r}-\sum\limits_{l=1}^{l'-1} \hat{g}_l\textbf{a}(\hat{d}_{l},\hat{\phi}_{l}).
\vspace{-0.25cm}
\end{equation}
We apply the ML estimates by minimizing the residual power  $\left\Arrowvert \textbf{r}_r - g\textbf{a}(d,\phi) \right\Arrowvert^2$, 
or equivalently, by maximizing $S(d,\phi,g)$, where
\vspace{-0.25cm}
\begin{equation}\label{ml5}
S(d,\phi,g) = 2\mathcal{R}\left\{ \textbf{r}_r^\text{H} g\textbf{a}(d,\phi) \right\} - 
\left\Arrowvert g\textbf{a}(d,\phi)  \right\Arrowvert^2.
\vspace{-0.25cm}
\end{equation}
The generalized likelihood ratio test estimate of $(d_{l'},\phi_{l'})$ of the $l'$-th path is the solution of the following optimization problem
\vspace{-0.25cm}
\begin{equation}\label{ml2}
(\hat{d}_{l'},\hat{\phi}_{l'})=\mathop{\arg\max}\limits_{ d\in \mathbb{\bar{S}}_d,\phi\in\mathbb{\bar{S}}_{\phi} }
|\textbf{a}(d,\phi)^\text{H} \textbf{r}_r|^2/
\left\Arrowvert  \textbf{a}(d,\phi) \right\Arrowvert^2.
\vspace{-0.25cm}
\end{equation}
The corresponding gain of the $l'$-th path that maximizes $S(d,\phi,g)$ is given by
\vspace{-0.25cm}
\begin{equation}\label{ml4}
\hat{g}_{l'} = \left(\textbf{a}(\hat{d}_{l'},\hat{\phi}_{l'})^\text{H} \textbf{r}_r\right)/
\left\Arrowvert  \textbf{a}(\hat{d}_{l'},\hat{\phi}_{l'}) \right\Arrowvert^2.
\vspace{-0.25cm}
\end{equation}

\subsubsection{Refinement stage}
Given that $d_{l'}$ and $\phi_{l'}$ can take any value in 
$\mathbb{R}$ and $(-{\pi}/{2},{\pi}/{2})$, respectively, we add a  refinement stage by utilizing Newton algorithm to reduce the off-grid effect and enhance the estimation accuracy.
Let $(\hat{d}_{l'},\hat{\phi}_{l'},\hat{g}_{l'})$ denote the current estimates.
The Newton refinement is given by
\vspace{-0.25cm}
\begin{equation}\label{newton1}
\begin{pmatrix} 
\hat{\hat{d}}_{l'} \\
\hat{\hat{\phi}}_{l'}
\end{pmatrix} = 
\begin{pmatrix} 
\hat{{d}}_{l'} \\
\hat{{\phi}}_{l'}
\end{pmatrix} 
-
\begin{pmatrix} 
\frac{\partial^2 S}{\partial d^2}   & \frac{\partial^2 S}{\partial d  \partial \phi} \\
\frac{\partial^2 S}{\partial  \phi \partial d } & \frac{\partial^2 S}{\partial  \phi^2 }
\end{pmatrix}^{-1} 
\begin{pmatrix} 
\frac{\partial S}{\partial d} \\
\frac{\partial S}{\partial \phi} 
\end{pmatrix}
\vspace{-0.25cm}
\end{equation} 
where the first-order partial derivatives of $S(d,\phi,g)$ is given by 
\vspace{-0.25cm}
\begin{equation}\label{newton2}
\frac{\partial S}{\partial x} = \mathcal{R}\left\{ (\textbf{r}_r - g\textbf{a}(d,\phi))^\text{H} g \frac{\partial \textbf{a}(d,\phi)}{\partial x} \right\},
\vspace{-0.25cm}
\end{equation} 
where $x$ can be $d$ or $\phi$.
The second-order partial derivatives of $S(d,\phi,g)$ is given by 
\vspace{-0.25cm}
\begin{equation}\label{newton3}
\frac{\partial^2 S}{\partial x \partial y} = \mathcal{R}\left\{ \left(\textbf{r}_r - g\textbf{a}(d,\phi)\right)^\text{H} g \frac{\partial^2 \textbf{a}(d,\phi)}{\partial x \partial y}
- |g|^2 \frac{\partial \textbf{a}^\text{H}(d,\phi)}{\partial x}\frac{\partial \textbf{a}(d,\phi)}{\partial y}
 \right\},
 \vspace{-0.25cm}
\end{equation} 
where $x$ and $y$ can be $d$ or $\phi$. 
Refer to \eqref{sv_simple}-\eqref{dwp} and some tedious calculations, \eqref{newton2} and \eqref{newton3} can be obtained.
The gain is then updated to  
\vspace{-0.25cm}
\begin{equation}\label{newton4}
\hat{\hat{g}}_{l'} = \left(\textbf{a}(\hat{\hat{d}}_{l'},\hat{\hat{\phi}}_{l'})^\text{H} \textbf{r}_r\right)/
\left\Arrowvert  \textbf{a}(\hat{\hat{d}}_{l'},\hat{\hat{\phi}}_{l'}) \right\Arrowvert^2.
\vspace{-0.25cm}
\end{equation} 
We accept a refinement only if the new residual power $\left\Arrowvert \textbf{r}_r - \hat{\hat{g}}_{l'}\textbf{a}(\hat{\hat{d}}_{l'},\hat{\hat{\phi}}_{l'}) \right\Arrowvert^2$ is smaller than the old residual power $\left\Arrowvert \textbf{r}_r - \hat{g}_{l'}\textbf{a}(\hat{d}_{l'},\hat{\phi}_{l'}) \right\Arrowvert^2$.

Note that, on the one hand, we can realize localization with the estimated location parameters;
on the other hand, we can design data transmission 
with the reconstructed the channel between the BS and the UE by using \eqref{ha}.
We can reconstruct the channel between the BS and all different UEs by orthogonal pilot signals.
The spectral efficiency based on the reconstructed channel is also analyzed in Section \ref{nr}.

\vspace{0.2cm}
\section{Multi-user Communication}
\begin{figure}[t]
	\vspace{-0.8cm}
	\centering
	\includegraphics[scale=0.62,angle=0]{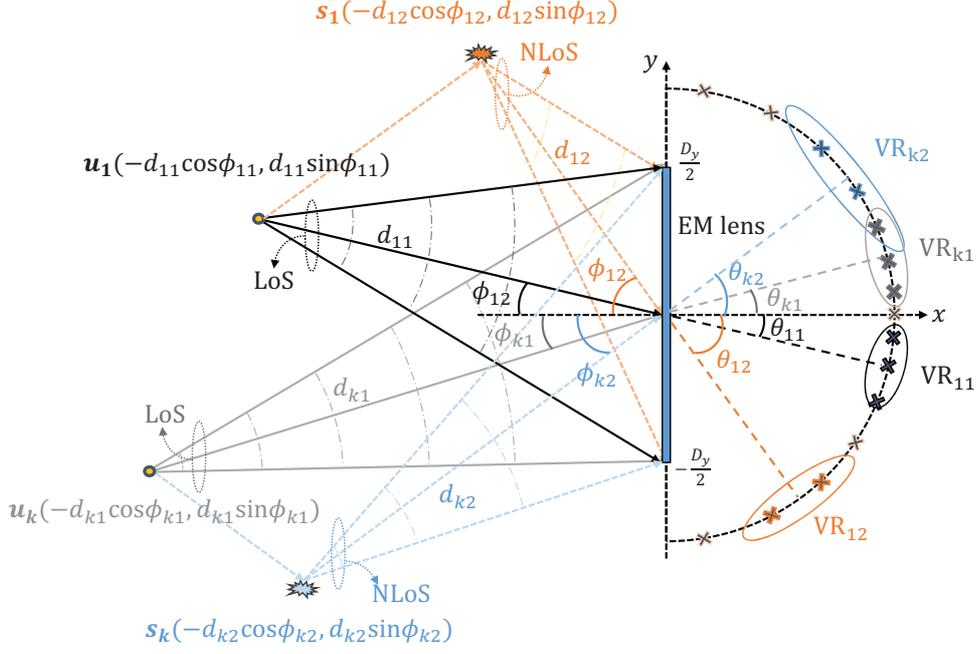}	
	\captionsetup{font=footnotesize}
	\caption{Multi-user communication with ExLens in spherical-wave scenarios.}
	\label{fig:multi-user}	
	\vspace{-0.6cm}
\end{figure}

The method proposed in Section \ref{ps} can obtain the location parameters together with the channel gains.
Multiple UEs can be simultaneously served by the ExLens system with the channels of all UEs reconstructed.
In this section, we analyze the multi-user communication performance of ExLens with limited RF chains in the coexistence of near-field and far-field UE.

We assume that $K$ single-antenna UEs are served by a BS, which is equipped with an ExLens with $N_a$ antenna elements (Fig. \ref{fig:multi-user}).
We consider the multi-path environment because of the presence of scatterers.
For illustration purposes, we only draw two paths for each UE (one line-of-sight (LoS) path and one non-line-of-sight (NLoS) path) in Fig. \ref{fig:multi-user} and omit the UE in the far-field.
Each path corresponds to an energy focusing window (FW) at the focal arc, and the overlap of the FWs of different UEs introduces inter-user-interference (IUI).
We consider the narrow band mmWave multi-path channel model \cite{hy}. 
According to \eqref{ha},
the channel between the BS and the $k$-th UE can be expressed as 
\vspace{-0.25cm}
\begin{equation}\label{h}
\textbf{h}_k = \sum\limits_{l=1}^{L_k} g_{kl}\textbf{a}(d_{kl},\phi_{kl}),
\vspace{-0.25cm}
\end{equation} 
where $L_k$ is the number of paths of the $k$-th UE, in which $l=1$ corresponds to the LoS path and $1<l\leq L_k$ corresponds to the NLoS path, $g_{kl}$ is the complex path gain for the $l$-th path of the $k$-th UE,
$(d_{k1},\phi_{k1})$ are the position parameters of the $k$-th UE,
$(d_{kl},\phi_{kl})$ are the position parameters of the $l$-th scatterer of the $k$-th UE with ($1<l\leq L_k$),
and $\textbf{a}(d_{kl},\phi_{kl})$ is the array response vector with elements given in \eqref{sv}.
Let $x_k=\sqrt{p_k}s_k$ represent the transmitted signal by the $k$-th UE, where $\sqrt{p_k}$ denotes the transmitted power, and $s_k$ denotes the independent information-bearing symbol with $\mathbb{E}\{|s_k|^2\}=1$.
The signal received at the BS is given as
\vspace{-0.25cm}
\begin{equation}\label{mu2}
\tilde{\textbf{r}} = \sum\limits_{k=1}^{K}\mathbf{h}_k{x_k} + \textbf{n} = \sum\limits_{k=1}^{K}\sum\limits_{l=1}^{L_k} g_{kl}\textbf{a}(d_{kl},\phi_{kl}){x_k} + \textbf{n},
\vspace{-0.25cm}
\end{equation}  
where $\textbf{n}\in \mathbb{C}^{N_a\times 1}$ represents the Gaussian noise with zero-mean and covariance matrix ${\sigma^2}\textbf{I}$.

We can reduce the number of RF chains for systems equipped with such array by exploiting the energy focusing property of the ExLens in near-field and far-field. 
We assume that only $M_{RF}$ RF chains are available, where $M_{RF}<N_a$.
Thus, antenna selection (e.g., by the low-complexity power-based antenna selection method) must be applied.
Let $\textbf{W}_{\text{RF}} \in \mathbb{R}^{N_a\times M_{RF}}$ denote the power-based antenna selection matrix, where the elements of $\textbf{W}_{\text{RF}}$ are $0$ or $1$.
To avoid the scenario in which antenna selection favors nearby UE over distant ones, we assume that the channel-inversion based power control is applied during the antenna selection phase.
Thus, the received signals at the BS from different UEs have comparable strength.
The signal received by the selected antennas can be expressed as 
\vspace{-0.05cm}
\begin{equation}\label{mu3}
\textbf{W}_{\text{RF}}^{\text{H}}\tilde{\textbf{r}} = \sum\limits_{k=1}^{K}\textbf{W}_{\text{RF}}^{\text{H}}\mathbf{h}_{k}{x_k} + \textbf{W}_{\text{RF}}^{\text{H}}\textbf{n}.
\vspace{-0.25cm}
\end{equation}
Given $\tilde{\textbf{r}}_s  \buildrel \Delta \over =  \textbf{W}_{\text{RF}}^{\text{H}}\tilde{\textbf{r}}$,  $\mathbf{h}_{k,s}  \buildrel \Delta \over =  \textbf{W}_{\text{RF}}^{\text{H}}\mathbf{h}_{k}$, and $\textbf{n}_s  \buildrel \Delta \over =  \textbf{W}_{\text{RF}}^{\text{H}}\textbf{n}$, we have
\vspace{-0.25cm}
\begin{equation}\label{mu33}
\tilde{\textbf{r}}_s =  \sum\limits_{k=1}^{K}\mathbf{h}_{k,s}{x_k} + \textbf{n}_s,
\vspace{-0.25cm}
\end{equation}
which can be rewritten as
\vspace{-0.05cm}
\begin{equation}\label{mu4}
\tilde{\textbf{r}}_s = \mathbf{h}_{k,s}{x_k} + \sum\limits_{k'\neq k}^{K}\mathbf{h}_{k',s}{x_{k'}} + \textbf{n}_s,
\vspace{-0.25cm}
\end{equation} 
where the term $\sum\limits_{k'\neq k}^{K}\mathbf{h}_{k',s}{x_{k'}}$ is the IUI for the $k$-th UE, and $\textbf{n}_s \in \mathbb{C}^{M_{RF}\times 1}$ represents the Gaussian noise at the selected antennas with zero-mean and covariance matrixc ${\sigma^2}\textbf{I}$.
Let $\textbf{u}_k \in \mathbb{C}^{M_{RF}\times 1}$ represent the baseband combining vector for the $k$-th UE, where $||\textbf{u}_k||=1$. The bandwidth-normalized achievable rate for the $k$-th UE is given by 
\vspace{-0.15cm}
\begin{equation}\label{rate1}
R_k = \log_2 \left(1+\dfrac{p_k|\textbf{u}_k^{\text{H}}\mathbf{h}_{k,s}|^2}{\sum\limits_{k'\neq k}^{K}{p_{k'}}|\textbf{u}_{k}^{\text{H}}\mathbf{h}_{k',s}|^2 + \sigma^2}\right),
\vspace{-0.25cm}
\end{equation} 
and for all the $K$ UE, we obtain the sum-rate as
\vspace{-0.15cm}
\begin{equation}\label{sum-rate}
R = \sum\limits_{k=1}^{K}R_k =  \sum\limits_{k=1}^{K}\log_2 \left(1+\dfrac{p_k|\textbf{u}_k^{\text{H}}\mathbf{h}_{k,s}|^2}{\sum\limits_{k'\neq k}^{K}{p_{k'}}|\textbf{u}_{k}^{\text{H}}\mathbf{h}_{k',s}|^2 + \sigma^2}\right).
\vspace{-0.25cm}
\end{equation} 
In near-field and far-field scenarios, the ExLens has the energy focusing ability. When the incident angles of different UEs are sufficiently separated, the ExLens can resolve various UEs.
For general systems where the BS cannot resolve all the UEs perfectly, we apply the linear receivers described in the following subsection to detect the signals from different UEs.

\subsection{Linear Receivers}

The combining vector $\textbf{u}_k$ is applied to \eqref{mu4} to detect $s_k$.
First, we consider the MRC scheme, which disregards the IUI term in \eqref{mu4}. In this case, $\textbf{u}_k$ is designed to simply maximize the desired signal power of the $k$-th UE, as given by
\vspace{-0.25cm}
\begin{equation}\label{mu5}
\textbf{u}_k^{*} = \mathop{\arg\max}_{\parallel \textbf{u}_k\parallel=1} \mid \textbf{u}_k^{\text{H}}\mathbf{h}_{k,s}\mid^2.
\vspace{-0.25cm}
\end{equation}
The optimal solution to \eqref{mu5} is
\vspace{-0.05cm}
\begin{equation}\label{mu6}
\textbf{u}_k^{*} = \dfrac{\mathbf{h}_{k,s}}{\parallel \mathbf{h}_{k,s} \parallel}.
\vspace{-0.25cm}
\end{equation}
The combining vector $\textbf{u}_k$ designed by the MRC scheme in \eqref{mu6} is sub-optimal in general because it ignores the IUI.

To further mitigate the IUI, we apply the MMSE-based combining scheme.
The MMSE considers the interference and finds the $\textbf{u}_k$, which minimizes the mean square error of the combined received and the desired signals, given as
\vspace{-0.25cm}
\begin{equation}\label{mu7}
\textbf{u}_k^{*} = \mathop{\arg\min}_{\parallel \textbf{u}_k\parallel=1} \mathbb{E}\{ \mid \textbf{u}_k^{\text{H}}\tilde{\mathbf{r}}_{s}-s_k\mid^2\}.
\vspace{-0.25cm}
\end{equation}
The optimal MMSE solution is 
\vspace{-0.05cm}
\begin{equation}\label{mu8}
\textbf{u}_k^{*} = \dfrac{\mathbf{R}_{rr}^{-1}\mathbf{R}_{rs}}{\parallel \mathbf{R}_{rr}^{-1}\mathbf{R}_{rs}\parallel},
\vspace{-0.25cm}
\end{equation}
where $\mathbf{R}_{rr}=\sum\limits_{k=1}^{K}p_k\mathbf{h}_{k,s}\mathbf{h}_{k,s}^{\text{H}}+\sigma^2\mathbf{I}$ denotes the autocorrelation matrix of received signal $\tilde{\mathbf{r}}_{s}$ and $\mathbf{R}_{rs}=\sqrt{p_k}\mathbf{h}_{k,s}$ denotes the cross-correlation of the received signal $\tilde{\mathbf{r}}_{s}$ and the desired signal $s_k$.

\subsection{Benchmark Schemes}

We compare the multi-user communication performance of ExLens with a conventional ULA in which both are illuminated by spherical wave-fronts.
We assume that both types of arrays have the same electrical aperture $\tilde{D}_y$ and number of antenna elements ($N_a = 2 \lfloor\tilde{D}_y\rfloor +1$).
Let ULA be placed along the y-axis centered at the origin, and the space between two adjacent antenna elements is $\Delta d = {\lambda}/{2}$ (Fig. \ref{fig:ULA}).
Take one UE for example, whose position parameters are $(d,\phi)$,
with $d$ denoting the distance between the UE and the original point, and $\phi \in (-\pi/2,\pi/2)$ as the angle of the UE relative to the x-axis. 
According to \cite{ULA2}, the array response of the ULA illuminated by the spherical wave-front is given by 
\vspace{-0.15cm}
\begin{equation}\label{ula}
\textbf{a}(d,\phi)=\left(\dfrac{d}{d_{_{-N}}}e^{-jk_0(d_{_{-N}}-d)},\ldots , \dfrac{d}{d_{_{-1}}}e^{-jk_0(d_{_{-1}}-d)}, 1, \dfrac{d}{d_{_{1}}}e^{-jk_0(d_{_{1}}-d)},\ldots,\dfrac{d}{d_{_{N}}}e^{-jk_0(d_{_{N}}-d)} \right),
\vspace{-0.15cm}
\end{equation}
where $d_n=\sqrt{d^2 + n^2\Delta d^2-2nd\Delta d\sin\phi}$ and $n \in \{0,\pm 1, \ldots,\pm N\}$.
When $d\to \infty$, ${d}/{d_n}\to 1$ and $ d_n - d  \to - n\Delta d\sin\phi$ are clearly observed.
Then, the array response illuminated by the spherical wave-front given in \eqref{ula} reduces to that illuminated by plane wave-front.   
Thus, the array response given in \eqref{ula} for ULA can be used for near-field and far-field scenarios.
We obtain the signal received by the ULA with a spherical wave-front by substituting \eqref{ula} into \eqref{mu2} with $K$ UE.
We assume that the ULA is equipped with $M_{RF}$ RF chains.
Thus, antenna selection is necessary.
However, given that the optimal antenna selection scheme for the ULA system illuminated by the spherical wave-front with multi-users is unknown in general, we apply two analog combining matrix design schemes as benchmarks for a remarkable comparison.
In the first benchmark,
we adopt the power-based antenna selection because of its simplicity.
Then, we apply the MRC and MMSE-based digital combining vector design schemes to the ULA system.
However, when $M_{RF}$ is small, the performance for ULA is limited because of the limited array gain with the small number of antennas selected.
Thus, we also consider the second benchmark by applying the approximate Gram Schmidt-based hybrid precoding scheme to design the analog combining matrix for ULA \cite{GS}.
To mitigate the IUI, the MMSE-based digital combining vector design scheme is then applied.

\section{Numerical Results}\label{nr}

\subsection{Localization Performance}

In this subsection, we discuss the performance of the proposed localization method.
We define $\mbox{CRLB}(\mathbf{d})=\sqrt{\sum_{k=1}^{L}{\bf{F}}^{-1}(\bm{\eta})_{(2k-1,2k-1)}}$ for comparison, where ${\bf{F}}^{-1}(\bm{\eta})_{(2k-1,2k-1)}$ denote the $(2k-1,2k-1)$-th element of the matrix ${\bf{F}}^{-1}(\bm{\eta})$.
Similarly,  we define $\mbox{CRLB}({\bm \phi})=\sqrt{\sum_{k=1}^{L}{\bf{F}}^{-1}(\bm{\eta})_{(2k,2k)}}$.
The PEB($\mathbf{u}$) is calculated according to \eqref{peb}.
Let $\mathbf{x}$ represent $\mathbf{d}$, ${\bm \phi}$, or $\mathbf{u}$, and $\hat{\mathbf{x}}_i$ denote
the estimate of $\mathbf{x}$ at the $i$-th Monte Carlo simulation.
The RMSE is defined as 
$\mbox{RMSE}(\mathbf{x})=\sqrt{\sum_{i=1}^{T}||\hat{\mathbf{x}}_i-\mathbf{x}||^2/T}$.
We define the normalized mean square error (NMSE) as $\sum_{i=1}^{T}||\mathbf{h}-\hat{\mathbf{h}}_i||^2/||\mathbf{h}||^2/T$ for channel estimation, where $\hat{\mathbf{h}}_i$ is the estimate of $\mathbf{h}$ at the $i$-th Monte Carlo simulation.
All numerical results provided in this subsection are obtained from $T=1000$ independent Monte Carlo simulations.
The position parameters $(d,\phi)$ of a UE is generated with $d \sim \mathcal{U}[7,30]\,m$ and $\phi \sim \mathcal{U}[-\pi/5,\pi/5]$ rad, where $\mathcal{U}$ denotes the uniform distribution.  
The settings for the ExLens are fixed to $D_y=1\,m$, $F_0=F=5\,m$, $\lambda = 0.01 \,m$, and $N_a = 201$.

Fig. \ref{Est} demonstrates that
(a) $\mbox{RMSE}$ versus $\mbox{CRLB}$ for the estimate of $\mathbf{d}$;
(b) $\mbox{RMSE}$ versus $\mbox{CRLB}$ for the estimate of ${\bm \phi}$;
(c) $\mbox{RMSE}$ versus $\mbox{PEB}$ for the estimate of $\mathbf{u}$;
and	
(d) $\mbox{NMSE}$ for the channel estimation.
First, we observe an improvement in the estimation accuracy of proposed method (denoted as NOMP) with respect to DOMP, where DOMP reaches performance floors with the increase in SNR. 
The performance plateau shows a fundamental algorithmic limitation of DOMP, and highlights the critical role of cyclic Newton refinements in NOMP, as explained in \cite{n1}.
Second, NOMP does not achieve the CRLB in the estimate of $\mathbf{d}$ (Fig. \ref{Est}(a)), but it closely follows the bound for all SNRs.  In the estimate of ${\bm \phi}$ (Fig. \ref{Est}(b)), NOMP can achieve the CRLB.
The array response of the ExLens is more sensitive to the changes in $\phi$ than that in $d$ in spherical scenarios.
Accordingly, NOMP performs better in ${\bm \phi}$ estimation than $\mathbf{d}$ estimation.
Third, the performance of $\mathbf{u}$ estimates (Fig. \ref{Est}(c)) shows similar trends to that of $\mathbf{d}$ (Fig. \ref{Est}(a)).
The estimation error of position $\mathbf{u}$ is mainly determined by the estimation error of $\mathbf{d}$ because the estimate of ${\bm \phi}$ achieves CRLB.
In the simulation settings, the proposed localization method can achieve meter-, decimeter-, and centimeter-level accuracies when SNR $> 0$, $> 20$, and $> 40$ dB, respectively.
Moreover, Fig. \ref{Est}(d) shows that the proposed method performs well in channel estimation, thereby demonstrating that the channel and location parameter estimation can be simultaneously performed in  spherical-wave scenarios by directly reusing
the communication signals.

\begin{figure}[H]
	\vspace{-0.5cm}
	\centering
	\includegraphics[scale=0.45,angle=0]{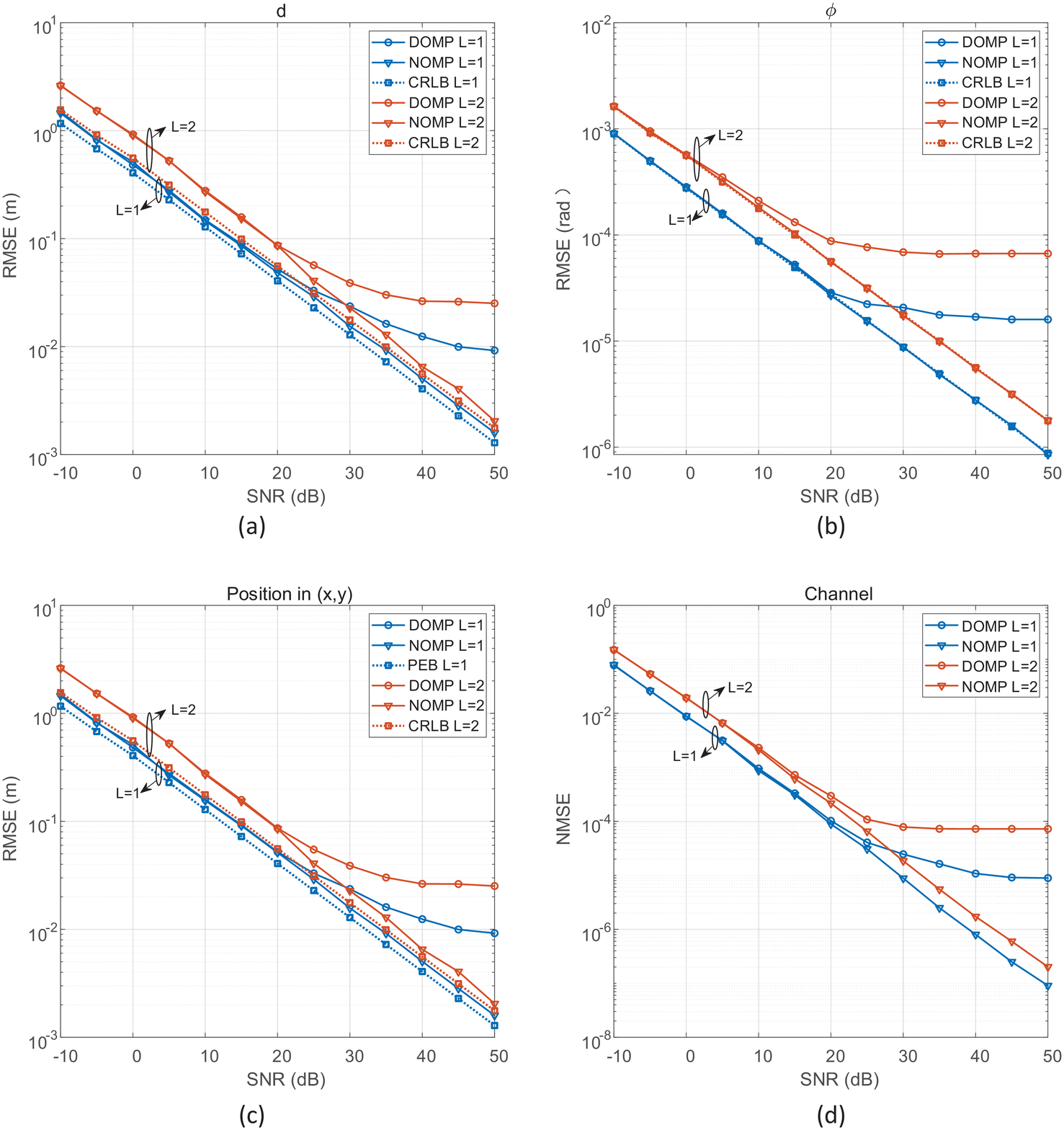}				
	\captionsetup{font=footnotesize}
	\caption{
		(a) $\mbox{RMSE}$ versus $\mbox{CRLB}$ for the estimate of $\mathbf{d}$.
		(b) $\mbox{RMSE}$ versus $\mbox{CRLB}$ for the estimate of ${\bm \phi}$.	
		(c) $\mbox{RMSE}$ versus $\mbox{PEB}$ for the estimate of $\mathbf{u}$.	
		(d) $\mbox{NMSE}$ for the estimate of $\mathbf{h}$.
		The settings for the ExLens are fixed to $D_y=1\,m$, $F_0=F=5\,m$, $\lambda = 0.01\,m$, and $N_a = 201$.
		In the case $L = 1$, $d = 16.8837 \,m$ and $\phi = 0.0693$ rad.
		In the case $L = 2$, $d_1 = 12.8657 \,m$, $\phi_1 = -0.1935$ rad, $d_2 = 14.4962 \,m$, and $\phi_2 = 0.1897$ rad.
	}\label{Est}
	\vspace{-0.6cm}
\end{figure}

\subsection{Multi-user Communication Performance}
In this subsection, we compare the multi-user communication performance of ExLens with that of the conventional ULA antenna array.
In the following simulations, the ExLens and ULA systems serve near-field and far-field UE simultaneously.
We assume that ExLens and ULA have the same electrical aperture $\tilde{D}_y$ and number of antenna elements $N_a = 2\lfloor\tilde{D}_y\rfloor+1$.
For the approximate Gram Schmidt-based hybrid precoding scheme applied to the benchmark ULA system, the size of the beamsteering codebook $N_{cb}$ is set to $1024$, and the resolution of the phase shifters in the analog combining network is assumed to be $10$ bits.
All numerical results provided in this section are obtained from Monte Carlo simulations with $1000$ independent channel realizations.
The position parameters $(d,\phi)$ of a UE is generated with $d \sim \mathcal{U}[20,320]\,m$ and $\phi \sim \mathcal{U}[-\pi/5,\pi/5]$ rad.
The low-complexity power-based antenna selection method is applied to ``LENS MMSE", ``LENS MRC", ``ULA MMSE", and ``ULA MRC", and the Gram Schmidt-based analog combining method is applied to ``ULA GS MMSE".

\begin{figure}[t]
	\centering
	\vspace{-0.8cm}	
	\begin{minipage}[t]{0.48\textwidth}
		\centering
		\includegraphics[scale=0.63,angle=0]{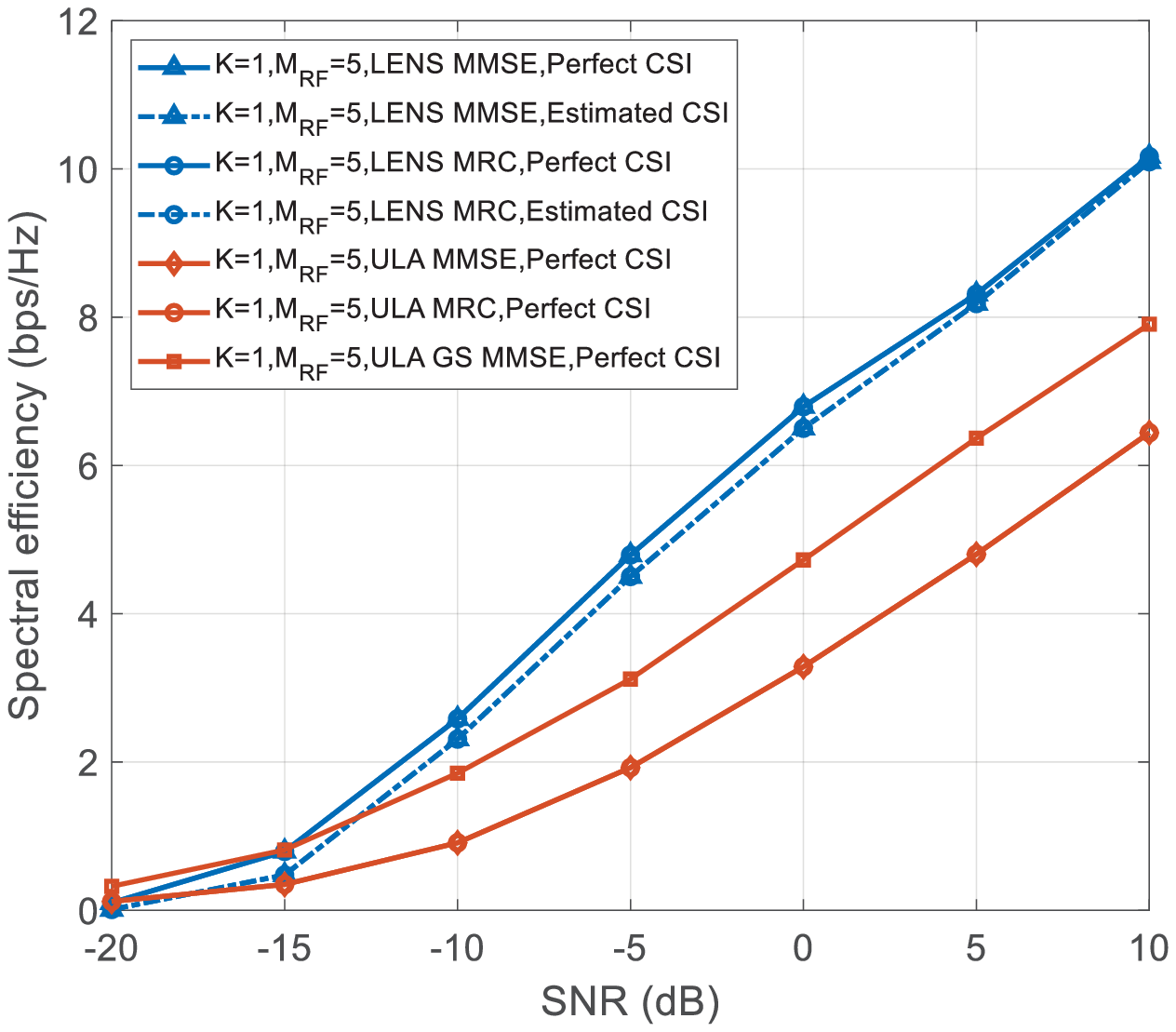}	
		\captionsetup{font=footnotesize}
		\caption{Comparison of the spectral efficiencies for single-user with different SNRs, where $K=1$, $M_{RF}=5$, $L=2$, $F=5\,m$, $F_0= 15\,m$, $\tilde{D}_y=100$, $\lambda = 0.01\,m$, and $N_a = 201$.}
		\label{fig:SNR1}
	\end{minipage}
	\hspace{0.25cm}
	\begin{minipage}[t]{0.48\textwidth}
		\centering
		\includegraphics[scale=0.63,angle=0]{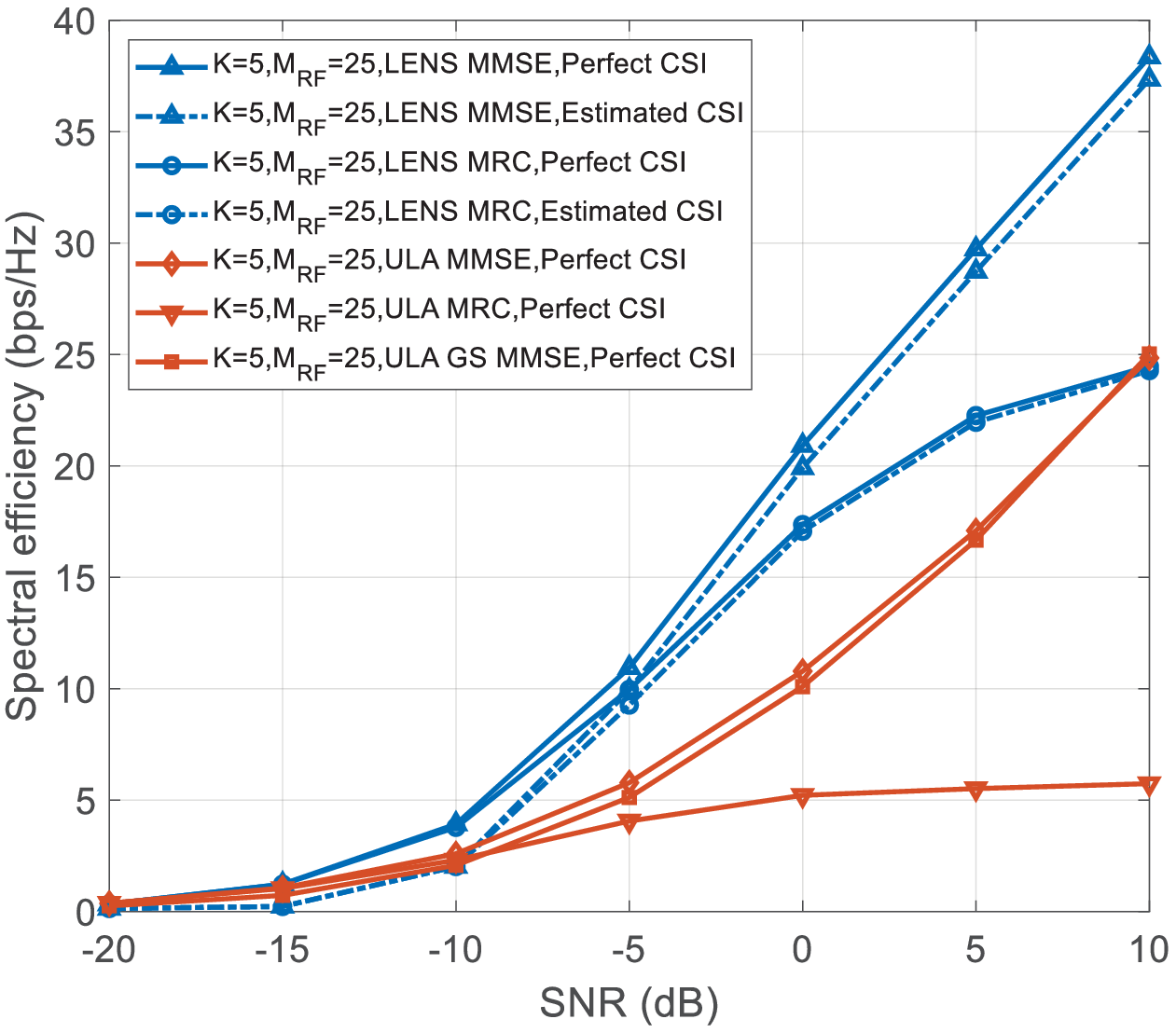}	
		\captionsetup{font=footnotesize}
		\caption{Comparison of the spectral efficiencies for multi-user scenarios with different SNRs, where $K=5$, $M_{RF}=25$, $L_k=2$, $F=5\,m$, $F_0= 15\,m$, $\tilde{D}_y=100$, $\lambda = 0.01\,m$, and $N_a = 201$.}
		\label{fig:SNR}
	\end{minipage}   
\end{figure}

Figs. \ref{fig:SNR1} and \ref{fig:SNR} compare the spectral efficiencies of different schemes with varying SNRs. The single-user (Fig. \ref{fig:SNR1} with $K=1$, $L_k = 2$, $M_{RF}=5$) and multi-user (Fig. \ref{fig:SNR} with $K=5$, $L_k = 2$, $M_{RF}=25$) scenarios are considered. The other parameters for ExLens are fixed to $F=5\,m$, $F_0= 15\,m$, $\tilde{D}_y=100$, $\lambda = 0.01\,m$, and $N_a = 201$.
In the benchmark ULA system, we assume that perfect CSI is available at the BS. 
In the ExLens system, we consider both cases with perfect and estimated CSIs.
In the ExLens system with a limited number of RF chains, we apply the power based antenna selection before channel estimation.
The channel is then estimated by the method proposed in Section \ref{Loc} with the received signal from the selected antennas.
First, the spectral efficiency of the ExLens systems outperforms the ULA systems because most energy of the received signal is concentrated on the selected antennas for the ExLens systems with the energy focusing property. By contrast, the energy in the ULA system is almost evenly spread across each antenna.
The simple power-based antenna selection method causes significant energy loss for the ULA systems, thereby resulting in poor performance in terms of spectral efficiency.
Second, in single-user scenarios (Fig. \ref{fig:SNR1}), the ``ULA GS MMSE" scheme outperforms the simple power-based antenna selection schemes ``ULA MMSE" and ``ULA MRC". 
This phenomenon is expected because the approximate Gram Schmidt-based hybrid precoding method considers the channel characteristic when it is applied to the ``ULA GS MMSE" scheme.
However, the performance of the ``ULA GS MMSE" scheme with much higher computational complexity is still worse than that of ``LENS MMSE" and ``LENS MRC" schemes.
Given the absence of the IUI for single-user scenarios, the MRC and MMSE schemes have the same performance.
In multi-user scenarios (Fig. \ref{fig:SNR}), the advantages of the ExLens systems are more pronounced over the ULA systems.
The MRC schemes perform worse than the MMSE schemes, especially for high SNRs, due to the presence of the IUI.
Lastly, the performance of the ``LENS MMSE" and ``LENS MRC" schemes with the estimated CSI is close to that based on perfect CSI, thereby  showing the effectiveness of the proposed channel estimation method.

\begin{figure}[t]
	\centering
\vspace{-0.8cm}	
	\begin{minipage}[t]{0.48\textwidth}
		\centering
		\includegraphics[scale=0.63,angle=0]{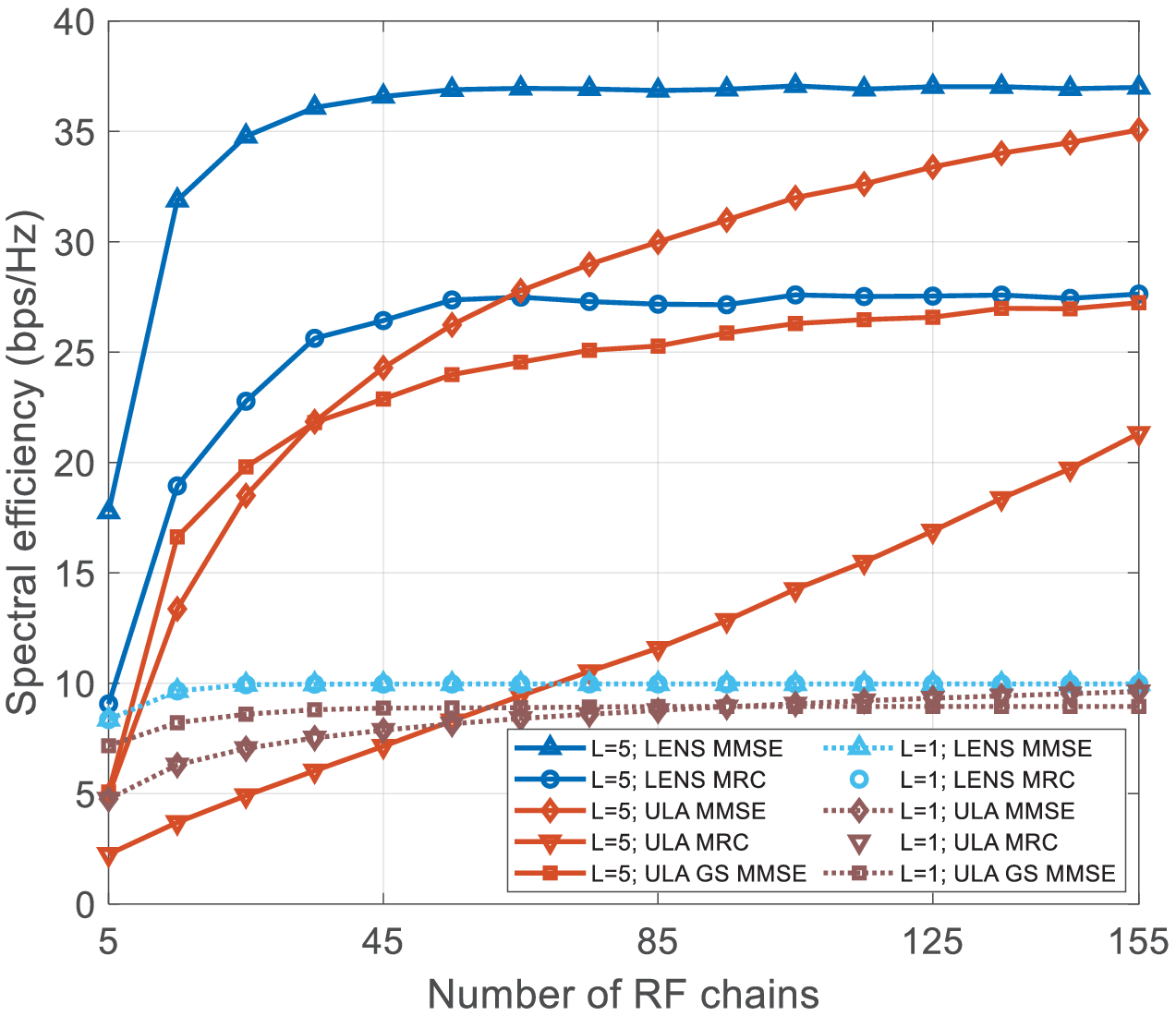}	
		\captionsetup{font=footnotesize}
		\caption{Comparison of the spectral efficiencies for single-user and multi-user scenarios with different number of RF chains, where $L_k=2$ for the $k$-th UE, $F=5\,m$, $F_0= 15\,m$, $\tilde{D}_y=100$, SNR$=10$ dB, $\lambda = 0.01\,m$, and $N_a = 201$.}
		\label{fig:MRF}
	\end{minipage}
	\hspace{0.25cm}
	\begin{minipage}[t]{0.48\textwidth}
		\centering
		\includegraphics[scale=0.63,angle=0]{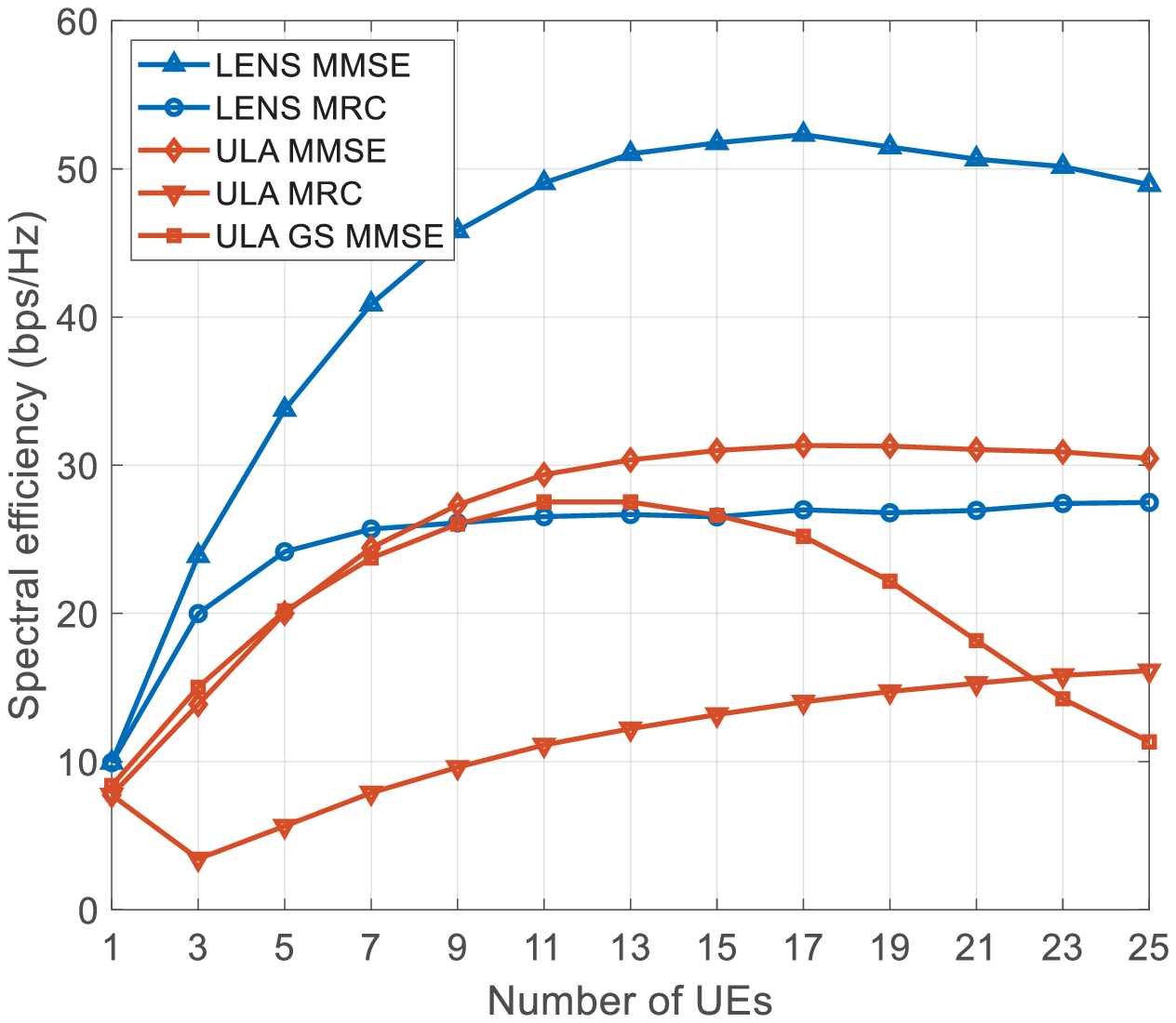}	
		\captionsetup{font=footnotesize}
		\caption{Comparison of the spectral efficiencies for multi-user scenarios with different number of UE, where $L_k=2$ for the $k$-th UE, $M_{RF}=20$, $F=5\,m$, $F_0= 15\,m$, $\tilde{D}_y=100$, SNR$=10$ dB, $\lambda = 0.01\,m$, and $N_a = 201$.}
		\label{fig:K}
	\end{minipage}   
	\vspace{-0.8cm}
\end{figure}

Then, we compare the spectral efficiencies of different schemes by increasing the number of RF chains for both single-user ($K=1$) and multi-user ($K=5$) scenarios (Fig. \ref{fig:MRF}).
As it benefits from the energy focusing property of ExLens,
the ``LENS MMSE" scheme always outperforms the ULA schemes for different numbers of RF chains.
For single-user scenarios,
as $M_{RF}$ increases, the spectral efficiencies of different schemes improve.
When $M_{RF}=15$, the performance of the ExLens schemes almost reaches maximum.
The ULA schemes require much more RF chains to achieve a similar performance.
Therefore, the energy focusing property of ExLens
is beneficial for reducing the number of RF chains, 
and this outcome helps to significantly reduce the signal processing complexity and hardware cost without notable performance degradation.
For the ULA schemes, ``ULA GS MMSE" shows advantages over other schemes when $M_{RF}$ is small, but with further increase in $M_{RF}$, the performance of ``ULA GS MMSE" saturates, a situation which is also explained in \cite{GS}. 
For multi-user scenarios, the number of RF chains required by the ``LENS MMSE" to achieve the optimal performance increases to around $45$ when $K$ increases to $5$.
Thus, more RF chains are needed to distinguish more UEs.
However, compared with the total number of antenna elements, the number of RF chains needed by the ExLens system remains low ($45 < 201$).    
The advantage of the ``LENS MMSE" scheme is more evident than the ``ULA MMSE" scheme with a smaller number of RF chains ($M_{RF}< 45$).

The performance of spectral efficiency versus the number of served UE for different schemes is shown in Fig. \ref{fig:K}, by fixing the number of the RF chains to $20$, the value of SNR to $10$ dB, and parameters for ExLens to $F=5\,m$, $F_0= 15\,m$, and $\tilde{D}_y=100$. 
The ``LENS MMSE" scheme always has the highest spectral efficiency among all others.
Therefore, the UE resolution of the ExLens system is greater than that of the ULA system with limited RF chains.
Since the IUI becomes larger as $K$ increases, the performance of the MRC schemes become worse than that of the MMSE schemes.
The beamsteering codebook for the approximate Gram Schmidt scheme is designed for the single-user systems of the ULA, 
where the analog combiner designed by the approximate Gram Schmidt scheme exhibits larger deviation from the real channel as $K$ increases.
Such higher deviation causes the worse performance of the ``ULA GS MMSE" than the ``ULA MMSE" when $K>5$. 
The spectral efficiency of the ``LENS MMSE" scheme initially increases with $K$, then decreases when $K>17$.
This trend is because the ability of the ExLens system to serve UE becomes limited given a number of RF chains.

\begin{figure}[t]
	\vspace{-0.8cm}
	\centering
	\includegraphics[scale=0.63,angle=0]{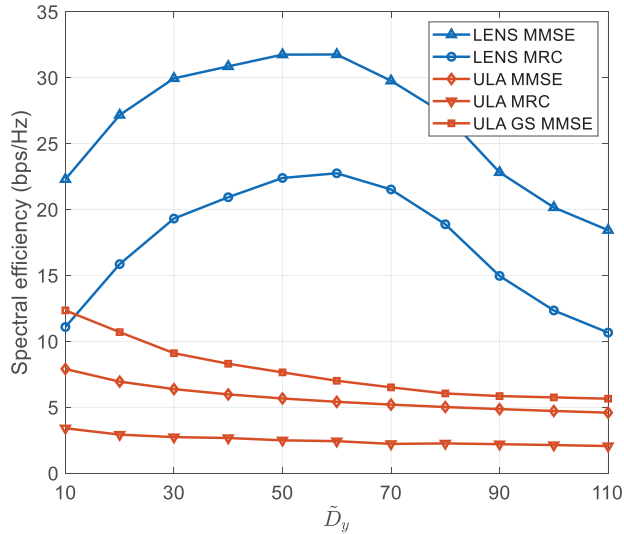}	
	\captionsetup{font=footnotesize}
	\caption{Comparison of the spectral efficiencies for multi-user scenarios with different antenna array aperture sizes, where $K=5$, $L_k=2$ for the $k$-th UE, $M_{RF}=5$, $F=5\,m$, $F_0= 15\,m$, SNR$=10$ dB, and $\lambda = 0.01\,m$.
}
	\label{fig:Dy}
	\vspace{-0.8cm}
\end{figure}

Finally, we evaluate the spectral efficiencies of different schemes by increasing the electrical aperture $\tilde{D}_y$ of the lens antenna array.
The results are shown in Fig. \ref{fig:Dy}, with $K=5$, $L_k=2$ for the $k$-th UE, $F=5\,m$, $F_0= 15\,m$, $M_{RF}=5$, and SNR$=10$ dB.
Channel inversion-based power control during the antenna selection phase is applied for all simulations. 
With unchanged total received power of all systems,
as $\tilde{D}_y$ increases, the total number of antenna elements increases for ULA systems, and the energy at each antenna element decreases.
Thus, with a fixed number of RF chains, the total received energy decreases accordingly, thereby leading to lower spectral efficiency for the ULA schemes.
However, ExLens systems present an interesting phenomenon, i.e., the spectral efficiency of the ExLens schemes shows a trend of first increasing and then decreasing.
According to the change of the energy focusing effects from far-field to near-field (Fig. \ref{fig:window2}), this phenomenon is easily understood.
When $\tilde{D}_y$ is very small, the ``sinc" function holds, and the width of the main lobe is $2/\tilde{D}_y$ and determines the system resolution to the UE. As $\tilde{D}_y$ increases, the system resolution rises.
Hence, the spectral efficiency of the lens systems increases with $\tilde{D}_y$ initially.
As $\tilde{D}_y$ further increases, the ``sinc" function no longer holds, and the near-field effect becomes obvious.
Then, the width of the focusing window determines the system resolution to the UE.
As $\tilde{D}_y$ further increases, the system resolution decreases. 
Thus, the spectral efficiency of the lens systems decreases with the further increase of $\tilde{D}_y$. 
During this process, the ExLens system resolution to the UE will reach the maximum at some value of $\tilde{D}_y$. Moreover, the optimal size of $\tilde{D}_y$ is $60$ under the simulation configuration given in Fig. \ref{fig:Dy}. These observations are instructive for the design of the electrical aperture size of the ExLens.

\section{Conclusion}
We considered the communication and localization problems with an ExLens.
First, we derived the closed-form antenna array response of ExLens by considering the spherical wave-front for two different EM lens designs.
The relationship between the antenna array response of ExLens in the near-field and far-field revealed that the derived near-field array response includes the existing ``sinc" function response as a special case.
We further analyzed the changes in the energy focusing properties from the far-field to the near-field and the difference of the energy focusing properties of the two EM lens designs.
The window focusing property in the near-field also revealed the great potential of ExLens for position sensing and multi-user communication.
The theoretical uplink localization ability of an ExLens was analyzed through the Fisher information.
To utilize the window focusing property for position sensing, 
an effective location parameter estimation method was next proposed.
The results showed that the localization performance is close to the CRLB and can be enhanced as the aperture of ExLens increase. 
In addition, the channel can be effectively reconstructed by the 
proposed estimation method.
Finally, the multi-user communication performance of ExLens that serves UE in near-field and far-field was investigated with perfect and  estimated CSIs.
Simulation results verified the effectiveness of the proposed channel estimation method and 
showed that the proposed ExLens with MMSE 
receiver achieves significant spectral efficiency gains and complexity-and-cost reductions compared with the ULA systems.

\begin{appendices}
\section{}\label{A}
In this section, we derive the array response of ExLens illuminated by a spherical wave-front.
For Design 1, 
by bringing \eqref{p1} into \eqref{r}, together with
$|| \mathbf{u} - \mathbf{p}|| = \sqrt {{d^2} + {y^2} - 2dy\sin \phi }$, $|| \mathbf{p}-\mathbf{b}_0|| = \sqrt {{F^2} + {y^2}}$, $|| \mathbf{p}-\mathbf{b}|| = \sqrt {{F^2} + {y^2}+2yF\sin \theta }$, $\eta (\mathbf{u}, \mathbf{p}) = {\lambda  /
		{\vphantom {\lambda  {( {2\pi \sqrt {{d^2} + {y^2} - 2dy\sin \phi } } )}}}  {( {4\pi \sqrt {{d^2} + {y^2} - 2dy\sin \phi } } )}}$,  and $\kappa(\mathbf{p},\mathbf{b})={\lambda}  / {(4\pi\!\sqrt {{F^2} + {y^2}+2yF\sin \theta })}$, we get
	\vspace{-0.25cm}
\begin{equation}\label{r1}
r(\theta ,d,\phi ) = \int\limits_{ - {{{D_y}} \mathord{\left/
			{\vphantom {{{D_y}} 2}} \right.
			\kern-\nulldelimiterspace} 2}}^{{{{D_y}} \mathord{\left/
			{\vphantom {{{D_y}} 2}} \right.
			\kern-\nulldelimiterspace} 2}} {\frac{\lambda^2 {e^{ - j{k_0}\sqrt {{d^2} + {y^2} - 2dy\sin \phi } }}{e^{ - j\left( {{\rm{ }}{\phi _0} + {k_0}\left( {\sqrt {{F^2} + {y^2} + 2yF\sin \theta } {\rm{ - }}\sqrt {{F^2} + {y^2}} } \right)} \right)}}}{{16\pi^2 \sqrt {{d^2} + {y^2} - 2dy\sin \phi }\ \sqrt {{F^2} + {y^2}+2yF\sin \theta } }}} dy.
		\vspace{-0.25cm}
\end{equation}
To derive the closed form of \eqref{r1}, we have to make the following assumptions: (A1) $d\gg y$ and (A2) $F\gg y$, where $y \in [-D_y/2,D_y/2]$.
Given (A1), we utilize Taylor series approximation and have
\vspace{-0.25cm} 
\begin{equation}\label{a1}    
\sqrt {{d^2} + {y^2} - 2dy\sin \phi }  \approx d - y\sin \phi  + {y^2}\frac{{{{\left( {\cos \phi } \right)}^2}}}{{2d}}.
\vspace{-0.25cm}
\end{equation}
Moreover, for the same reason, we have 
\vspace{-0.25cm}
\begin{equation}\label{a2}
\frac{1}{{\sqrt {{d^2} + {y^2} - 2dy\sin \phi } }} = \frac{1}{d}{\left( {1 + \frac{{{y^2}}}{{{d^2}}} - \frac{{2y}}{d}\sin \phi } \right)^{ - \frac{1}{2}}} \approx \frac{1}{d}. 
\vspace{-0.25cm}
\end{equation}
Then, given (A2), we also utilize Taylor series approximation and obtain 
\vspace{-0.25cm}
\begin{equation}\label{a3}    
\sqrt {{F^2} + {y^2} + 2yF\sin \theta } {\rm{ - }}\sqrt {{F^2} + {y^2}}  \approx \sqrt {{F^2} + {y^2}}  \times \left[ {\frac{{yF\sin \theta }}{{{F^2} + {y^2}}}{\rm{ - }}\frac{{\rm{1}}}{{\rm{2}}}{{\left( {\frac{{yF\sin \theta }}{{{F^2} + {y^2}}}} \right)}^{\rm{2}}}} \right],
\vspace{-0.25cm}
\end{equation}
for further simplification, we obtain
\vspace{-0.25cm}
\begin{equation}\label{a33}    
\sqrt {{F^2} + {y^2} + 2yF\sin \theta } {\rm{ - }}\sqrt {{F^2} + {y^2}}  \approx  y\sin \theta  - \frac{{{{\left( {y\sin \theta } \right)}^2}}}{{2F}}.
\vspace{-0.25cm}
\end{equation}
Moreover, for the same reason, we have 
\vspace{-0.25cm}
\begin{equation}\label{a333}
\frac{1}{\sqrt {{F^2} + {y^2}+2yF\sin \theta }} = \frac{1}{F}{\left( {1 + \frac{{{y^2}}}{{{F^2}}} - \frac{{2y}}{F}\sin \theta } \right)^{ - \frac{1}{2}}} \approx \frac{1}{F}. 
\vspace{-0.25cm}
\end{equation}
Substituting \eqref{a1}-\eqref{a333} into \eqref{r1}, we have 
\vspace{-0.25cm}
\begin{equation}\label{a4} 
r(\theta ,d,\phi ) \approx \int\limits_{ - {{{D_y}} \mathord{\left/
			{\vphantom {{{D_y}} 2}} \right.
			\kern-\nulldelimiterspace} 2}}^{{{{D_y}} \mathord{\left/
			{\vphantom {{{D_y}} 2}} \right.
			\kern-\nulldelimiterspace} 2}} \frac{\lambda^2 }{16\pi^2dF}{e^{ - j{k_0}\left( {d - y\sin \phi  + {y^2}\frac{ \cos ^2\phi }{2d}} \right)}}
		e^{ - j\left(\phi _0 + k_0 y\sin\theta-\frac{k_0y^2\sin^2 \theta  }{2F} \right)} dy.
		\vspace{-0.25cm}
\end{equation}
Rewritten \eqref{a4}, we get
\vspace{-0.25cm}
\begin{equation}\label{a5} 
r(\theta ,d,\phi ) \approx \frac{\lambda^2 {e^{ - j\left( {{k_0}d + {\phi _0}} \right)}}}{16\pi^2dF} \int\limits_{ - {{{D_y}} \mathord{\left/
			{\vphantom {{{D_y}} 2}} \right.
			\kern-\nulldelimiterspace} 2}}^{{{{D_y}} \mathord{\left/
			{\vphantom {{{D_y}} 2}} \right.
			\kern-\nulldelimiterspace} 2}} {{e^{j{y^2}{k_0}\left( {\frac{{{{\sin }^{\rm{2}}}\theta }}{{{\rm{2}}F}} - \frac{{{{\cos }^2}\phi }}{{2d}}} \right)}}{e^{ - jy{k_0}\left( { \sin\theta   - \sin\phi } \right)}}} dy.
		\vspace{-0.25cm}
\end{equation}
Without loss of generality, we assume $\phi_0=2\pi$ for the first lens design. Since $\phi_0$ is common for all antenna elements, the phase term $e^{ - j\phi_0}$ can be ignored.
Denote $\alpha  = \frac{{\pi \sin^2 \theta }}{{\lambda F}} - \frac{{\pi \cos^2 \phi }}{{\lambda d}}$, and $\beta  =  ({{ \sin \theta - \sin \phi}})/{\lambda }$, we obtain 
\vspace{-0.25cm}
\begin{equation}\label{a6} 
r(\theta ,d,\phi ) \approx \frac{\lambda^2 {e^{ - j {{k_0}d} }}}{16\pi^2dF}\int\limits_{ - {{{D_y}} \mathord{\left/
			{\vphantom {{{D_y}} 2}} \right.
			\kern-\nulldelimiterspace} 2}}^{{{{D_y}} \mathord{\left/
			{\vphantom {{{D_y}} 2}} \right.
			\kern-\nulldelimiterspace} 2}} {{e^{j\alpha {y^2}}}{e^{ - j2\pi \beta y}}} dy.
		\vspace{-0.25cm}
\end{equation}

For Design 2, 
with $\lVert{\mathbf{c}_0}-\mathbf{p}\rVert=\sqrt {{F_0}^2 + {y^2}}$, we have 
\vspace{-0.25cm}
\begin{equation}\label{r2}
r(\theta ,d,\phi ) = \int\limits_{ - {{{D_y}} \mathord{\left/
			{\vphantom {{{D_y}} 2}} \right.
			\kern-\nulldelimiterspace} 2}}^{{{{D_y}} \mathord{\left/
			{\vphantom {{{D_y}} 2}} \right.
			\kern-\nulldelimiterspace} 2}} {\frac{\lambda^2 {e^{ - j{k_0}(\sqrt {{d^2} + {y^2} - 2dy\sin \phi }-\sqrt {{F_0}^2 + {y^2}}) }}{e^{ - j\left( {{\rm{ }}{\phi _0} + {k_0}\left( {\sqrt {{F^2} + {y^2} + 2yF\sin \theta } {\rm{ - }}\sqrt {{F^2} + {y^2}} } \right)} \right)}}}{{16\pi^2 \sqrt {{d^2} + {y^2} - 2dy\sin \phi }\ \sqrt {{F^2} + {y^2}+2yF\sin \theta } }}} dy.
		\vspace{-0.25cm}
\end{equation}
Similarly, without loss of generality, we assume $\phi_0-k_0F_0=2\pi$ for the second lens design. Since $\phi_0-k_0F_0$ is common for all antenna elements, the phase term $e^{ - j(\phi_0-k_0F_0)}$ can be ignored.
The received signal can have the same approximate expression as \eqref{a6} with $\alpha = \frac{{\pi {{\sin }^{\rm{2}}}\theta }}{{\lambda F}} - \frac{{\pi {{\cos }^{\rm{2}}}\phi }}{{\lambda d}} + \frac{\pi }{{\lambda {F_0}}}$, as summarized in Table \ref{af}.
Let ${J_a} = \int\limits_{ - {{{D_y}} \mathord{\left/
			{\vphantom {{{D_y}} 2}} \right.
			\kern-\nulldelimiterspace} 2}}^{{{{D_y}} \mathord{\left/
			{\vphantom {{{D_y}} 2}} \right.
			\kern-\nulldelimiterspace} 2}} {{e^{j\alpha {y^2}}}{e^{ - j2\pi \beta y}}dy}$, 
		we have
		\vspace{-0.25cm}
\begin{equation}\label{a8}	
{J_a} = \frac{{\sqrt \pi  }}{{{\rm{2}}\sqrt \alpha  }}{e^{ - j\left( {\frac{{{{\left( {2\pi \beta } \right)}^2}}}{{4\alpha }} - \frac{5\pi }{4}} \right)}}\left( {\mathrm{erf}\left( {\frac{{\alpha {D_y} + 2\pi \beta }}{{2\sqrt \alpha  }}{e^{j\frac{{3\pi }}{4}}}} \right) + \mathrm{erf}\left( {\frac{{\alpha {D_y} - 2\pi \beta }}{{2\sqrt \alpha  }}{e^{j\frac{{3\pi }}{4}}}} \right)} \right).
\vspace{-0.25cm}
\end{equation}	
Hence, we obtain
\vspace{-0.05cm}
\begin{equation}\label{a11}	
r(\theta ,d,\phi) \approx \frac{{\lambda^2 {e^{ - j{k_0d}}} }}{{16\pi^2 dF}}J_a.
\vspace{-0.25cm}
\end{equation}
Then, we define that the effective lens antenna array response
on point $\mathbf{b}=[F\cos\theta,-F\sin\theta]$ at the focal arc as $a(\theta,d,\phi)=r(\theta,d,\phi)\times{16\pi^2 dF}/({{\lambda^2 {e^{ - j{k_0 d}}} }})$.
It then follows from \eqref{a11} that we have 
\vspace{-0.25cm}
\begin{equation}\label{a12}	
a(\theta,d,\phi) \approx \frac{{\sqrt \pi  }}{{{\rm{2}}\sqrt \alpha  }}{e^{ - j\left( {\frac{{{{\left( {2\pi \beta } \right)}^2}}}{{4\alpha }} - \frac{5\pi }{4}} \right)}}\left( {\mathrm{erf}\left( {\frac{{\alpha {D_y} + 2\pi \beta }}{{2\sqrt \alpha  }}{e^{j\frac{{3\pi }}{4}}}} \right) + \mathrm{erf}\left( {\frac{{\alpha {D_y} - 2\pi \beta }}{{2\sqrt \alpha  }}{e^{j\frac{{3\pi }}{4}}}} \right)} \right),
\vspace{-0.25cm}
\end{equation}	
where $\beta={{(\sin \theta  - \sin \phi) }}/{\lambda }$ and $\alpha$ is given in Table \ref{af} for different lens designs.

\section{}\label{B}
In this section, we give the proof of Lemma 1.	
By substituting the definition of $\mathrm{erf}(x)$ given in \eqref{erf} into \eqref{response}, and after some manipulations, we have 
\vspace{-0.25cm}
	\begin{equation}\label{plane11}
	a(\theta,d,\phi )  =  - \left.{{{e^{j\frac{\pi }{4}}}\left( {\int\limits_0^{\left( {\frac{{{D_y}\sqrt \alpha }}{2} + \frac{{\pi \beta}}{{\sqrt \alpha }}} \right){e^{j\frac{{3\pi }}{4}}}} {{e^{ - {t^2}}}dt}  + \int\limits_0^{\left( {\frac{{{D_y}\sqrt \alpha }}{2} - \frac{{\pi \beta}}{{\sqrt \alpha }}} \right){e^{j\frac{{3\pi }}{4}}}} {{e^{ - {t^2}}}dt} } \right)}}\middle/ {{(\sqrt \alpha {e^{j\frac{{{{\left( {\pi \beta} \right)}^2}}}{\alpha}}})}}\right..
	\vspace{-0.25cm}
	\end{equation}	
The assumption of plane wave-front holds when $d \to \infty$ and $F_0 \to \infty$ (for the second lens design), and also with the assumption that $F\gg y$, we can assume that in the far-field $\alpha\rightarrow 0$. Let $x\triangleq \sqrt \alpha$, we have	
\vspace{-0.25cm}
	\begin{equation}\label{plane2}
	\lim \limits_{d, F_0 \to \infty}a(\theta,d,\phi ) =	\mathop {\lim }\limits_{x \to 0} {\rm{ - }}\left.{{{e^{j\frac{\pi }{4}}}\left( {\int\limits_0^{\left( {\frac{{{D_y}x}}{2} + \frac{{\pi \beta}}{x}} \right){e^{j\frac{{3\pi }}{4}}}} {{e^{ - {t^2}}}dt}  + \int\limits_0^{\left( {\frac{{{D_y}x}}{2} - \frac{{\pi \beta}}{x}} \right){e^{j\frac{{3\pi }}{4}}}} {{e^{ - {t^2}}}dt} } \right)}}\middle/{{x{e^{j\frac{{{{\left( {\pi \beta} \right)}^2}}}{{{x^2}}}}}}}\right..
	\vspace{-0.25cm}
	\end{equation}
	Utilizing the L'Hopital's rule, \eqref{plane2} can be simplified to
	\vspace{-0.25cm}
	\begin{equation}\label{plane3}
	\begin{array}{ll}
	&\mathop {\lim }\limits_{x \to 0}  - \frac{{{e^{j\pi }}\left( {\left( {\frac{{{D_y}}}{2} - \frac{{\pi \beta}}{{{x^2}}}} \right){e^{j\left( {\frac{{{D_y}^2{x^2}}}{4} + \frac{{{{\left( {\pi \beta} \right)}^2}}}{{{x^2}}} + {D_y}\pi \beta} \right)}} + \left( {\frac{{{D_y}}}{2} + \frac{{\pi \beta}}{{{x^2}}}} \right){e^{j\left( {\frac{{{D_y}^2{x^2}}}{4} + \frac{{{{\left( {\pi \beta} \right)}^2}}}{{{x^2}}} - {D_y}\pi \beta} \right)}}} \right)}}{{{e^{j\frac{{{{\left( {\pi \beta} \right)}^2}}}{{{x^2}}}}} - 2j\frac{{{{\left( {\pi \beta} \right)}^2}}}{{{x^2}}}{e^{j\frac{{{{\left( {\pi \beta} \right)}^2}}}{{{x^2}}}}}}}\\
	= &\mathop {\lim }\limits_{x \to 0}  - \frac{{{e^{j\pi }}\left( {\left( {\frac{{{D_y}}}{2}{x^2} - \pi \beta} \right){e^{j\left( {\frac{{{D_y}^2{x^2}}}{4} + {D_y}\pi \beta} \right)}} + \left( {\frac{{{D_y}}}{2}{x^2} + \pi \beta} \right){e^{j\left( {\frac{{{D_y}^2{x^2}}}{4} - {D_y}\pi \beta} \right)}}} \right)}}{{{x^2} - 2j{{\left( {\pi \beta} \right)}^2}}},
	\end{array}
	\vspace{-0.25cm}
	\end{equation}	
	and with ${x \to 0}$, we have 
	\vspace{-0.25cm}
	\begin{equation}\label{plane33}
	\mathop {\lim }\limits_{x \to 0}  - \frac{{{e^{j\pi }}\left( {\left( {\frac{{{D_y}}}{2}{x^2} - \pi \beta} \right){e^{j\left( {\frac{{{D_y}^2{x^2}}}{4} + {D_y}\pi \beta} \right)}} + \left( {\frac{{{D_y}}}{2}{x^2} + \pi \beta} \right){e^{j\left( {\frac{{{D_y}^2{x^2}}}{4} - {D_y}\pi \beta} \right)}}} \right)}}{{{x^2} - 2j{{\left( {\pi \beta} \right)}^2}}}= \frac{{{e^{j{D_y}\pi \beta}} - {e^{ - j{D_y}\pi \beta}}}}{{2j\pi \beta}},
	\vspace{-0.25cm}
	\end{equation}
	since $({{{e^{j{D_y}\pi \beta}} - {e^{ - j{D_y}\pi \beta}}}})/({{2j\pi \beta}}) = {D_y}\mathrm{sinc}\left({{{{\tilde{D}_y}}}\sin \theta  - {{{\tilde{D}_y}}}\sin \phi } \right)$,	
	obviously, we have 
	\vspace{-0.25cm}
	\begin{equation}\label{plane333}
	\lim \limits_{d, F_0 \to \infty}a(\theta,d,\phi ) ={D_y}\mathrm{sinc}\left({{{{\tilde{D}_y}}}\sin \theta  -{{{\tilde{D}_y}}}\sin \phi } \right).
	\vspace{-0.25cm}
	\end{equation}

\section{}\label{C}
In this section, we give the proof of Lemma 2.
To analyze the property of ${w(\theta ,d,\phi)}$ in \eqref{w1}, we treat $\alpha$ and $\beta$ as a continuous function of $\theta$.
We firstly review the property of the $\mathrm{erf}(x)$ defined in \eqref{erf}, 
where $\mathrm{erf}(0)=0$, $\mathrm{erf}(\infty)=1$, and $\mathrm{erf}(-x)=-\mathrm{erf}(x)$.
According to \eqref{w1}, ${w(\theta ,d,\phi)}$ is the sum of two $\mathrm{erf}$ functions, and the amplitude of ${w(\theta ,d,\phi)}$ is shown in Fig. \ref{fig:window2}, which is similar to a rectangular window function in the near-field. The edges of the window are determined by zero points of two $\mathrm{erf}$ functions, namely $v_1$ and $v_2$.
Take the second lens design as an example, a given received waveform is shown in the last two subfigures of Fig. \ref{fig:window2}, where $v_1$ is obtained when $\xi_1=0$, and similarly $v_2$ is obtained when $\xi_2=0$. Note that $\xi_1 = {\frac{{\alpha{D_y} + 2\pi \beta}}{{2\sqrt \alpha }}{e^{j\frac{{3\pi }}{4}}}}$ and $\xi_2 = {\frac{{\alpha{D_y} - 2\pi \beta}}{{2\sqrt \alpha }}{e^{j\frac{{3\pi }}{4}}}}$.
Denote $v\buildrel \Delta \over = \sin\theta$, where $v\in(-1,1)$.
Let $\xi_1 =0$, 
we have 
\vspace{-0.25cm}
\begin{equation}\label{w2}	
{v^2} + \frac{{2F}}{{{D_y}}}v + \left( {\frac{F}{{{F_0}}} - \frac{{2F\sin \phi }}{{{D_y}}} - \frac{{F{{\cos }^2}\phi }}{d}} \right) = 0.
\vspace{-0.25cm}
\end{equation}
Since $v\in(-1,1)$, the only solution of \eqref{w2} is
\vspace{-0.25cm}
\begin{equation}\label{w3}	
{v_1} = -{{\frac{{F}}{{{D_y}}} + \sqrt {{{\left( {\frac{{F}}{{{D_y}}}} \right)}^2} - \left( {\frac{F}{{{F_0}}} - \frac{{2F\sin \phi }}{{{D_y}}} - \frac{{F{{\cos }^2}\phi }}{d}} \right)} }}.
\vspace{-0.25cm}
\end{equation}
Similarly,
let $\xi_2 =0$,
we have
\vspace{-0.25cm}
\begin{equation}\label{w4}	
{v^2} - \frac{{2F}}{{{D_y}}}v + \left( {\frac{F}{{{F_0}}} + \frac{{2F\sin \phi }}{{{D_y}}} - \frac{{F{{\cos }^2}\phi }}{d}} \right) = 0.
\vspace{-0.25cm}
\end{equation}
Since $v\in(-1,1)$, the only solution of \eqref{w4} is
\vspace{-0.25cm}
\begin{equation}\label{w5}	
{v_{\rm{2}}} = {{  \frac{{F}}{{{D_y}}}{\rm{ - }}\sqrt {{{\left( {\frac{{F}}{{{D_y}}}} \right)}^2} - \left( {\frac{F}{{{F_0}}}{\rm{ + }}\frac{{2F\sin \phi }}{{{D_y}}} - \frac{{F{{\cos }^2}\phi }}{d}} \right)} }}.
\vspace{-0.25cm}
\end{equation}
According to \eqref{w3} and \eqref{w5}, we can further obtain the center and width of the focusing window.
Let $v_c$ denote the center of the focusing window, we have
\vspace{-0.25cm}
\begin{equation}\label{w6}	
{v_c} = \frac{{{v_1} + {v_2}}}{2}
 = \frac{{  \frac{{16F\sin \phi }}{{{D_y}}}}}{{\frac{{8F}}{{{D_y}}}\left( {\sqrt {1 - \left( {\frac{{{D_y}^2}}{{F{F_0}}} + \frac{{2{D_y}\sin \phi }}{F} - \frac{{{D_y}^2{{\cos }^2}\phi }}{{Fd}}} \right)}  + \sqrt {1 - \left( {\frac{{{D_y}^2}}{{F{F_0}}} - \frac{{2{D_y}\sin \phi }}{F} - \frac{{{D_y}^2{{\cos }^2}\phi }}{{Fd}}} \right)} } \right)}},
 \vspace{-0.25cm}
\end{equation}
since $F,F_0,d\gg D_y$, we obtain
\vspace{-0.25cm}
\begin{equation}\label{w66}	
{v_c} \approx \left.\left({{ \frac{{16F\sin \phi }}{{{D_y}}}}}\right)\middle /\left({{\frac{{16F}}{{{D_y}}}}}\right) \right. =  \sin \phi .
\vspace{-0.25cm}
\end{equation}
Let $\Delta v$ denote the width of the focusing window, we have
\vspace{-0.25cm}
\begin{equation}\label{w7}
\hspace{-0.25cm}	
\begin{array}{ll}
\Delta v \!\!\!\!\! &=\left| {{v_{\rm{1}}} - {v_{\rm{2}}}} \right|\\
&= {{\left| {\frac{{F}}{{{D_y}}}\left( {{\rm{1}}\! - \!\sqrt {{\rm{1}} - \left( {\frac{{{D_y}^2}}{{F{F_0}}} + \frac{{{\rm{2}}{D_y}\sin \phi }}{F} - \frac{{{D_y}^2{{\cos }^2}\phi }}{{Fd}}} \right)} } \right){\rm{ + }}\frac{{F}}{{{D_y}}}\left( {{\rm{1}} \! -\! \sqrt {{\rm{1}} - \left( {\frac{{{D_y}^2}}{{F{F_0}}} \! -\! \frac{{{\rm{2}}{D_y}\sin \phi }}{F} - \frac{{{D_y}^2{{\cos }^2}\phi }}{{Fd}}} \right)} } \right)} \right|}},
\end{array}
\vspace{-0.25cm}
\end{equation}
similarly, since $F,F_0,d\gg D_y$, we get
\vspace{-0.25cm}
\begin{equation}\label{w77}	
\begin{array}{ll}
\Delta v
 & \approx \dfrac{{{\left| {\frac{F}{{{D_y}}}\left( {\frac{{{D_y}^2}}{{F{F_0}}} + \frac{{{\rm{2}}{D_y}\sin \phi }}{F} - \frac{{{D_y}^2{{\cos }^2}\phi }}{{Fd}}} \right){\rm{ + }}\frac{F}{{{D_y}}}\left( {\frac{{{D_y}^2}}{{F{F_0}}} - \frac{{{\rm{2}}{D_y}\sin \phi }}{F} - \frac{{{D_y}^2{{\cos }^2}\phi }}{{Fd}}} \right)} \right|}}}{2}\\
&= {D_y}\left| {\dfrac{1}{{{F_0}}} - \dfrac{{{{\cos }^2}\phi }}{d}} \right|.
\end{array}
\vspace{-0.25cm}
\end{equation}
The same procedure can be applied to obtain the approximate center and width of the focusing window for the first lens design.

\vspace{-0.25cm}
\section{}\label{FIM}
For the $n$-th element in $\mathbf{a}(d_l,\phi_l)$, we have
\vspace{-0.25cm}
\begin{equation}\label{sv_simple}
a_n(d_l,\phi_l) = am_n(d_l,\phi_l)\times ph_n(d_l,\phi_l)\times w_n(d_l,\phi_l),
\vspace{-0.25cm}
\end{equation}
where $am_n(d_l,\phi_l)=
\frac{{\sqrt \pi  }}{{{\rm{2}}\sqrt \alpha }}$,
$ph_n(d_l,\phi_l) = 
e^{ - j\left( {\frac{{{{ {\pi^2 \beta^2} }}}}{{\alpha}} - \frac{5\pi }{4}} \right)}$, and $w_n(d_l,\phi_l)$ is the discrete ``window" function derived from \eqref{w1} by replacing $\sin \theta$ with $\sin \theta_n = {n}/{N}$ in $\alpha$ and $\beta$ for $n \in \{0,\pm 1, \ldots,\pm N\}$.
We simplify $am_n(d_l,\phi_l)$, $ph_n(d_l,\phi_l)$ and $w_n(d_l,\phi_l)$ as $am_n$, $ph_n$ and $w_n$.
Then, we have
\vspace{-0.25cm}
\begin{equation}\label{dad}
\dfrac{{\partial a_n(d_l,\phi_l)}}{{\partial d_l}}
= \dfrac{{\partial am_n}}{{\partial d_l}}\times ph_n\times w_n
+am_n\times \dfrac{{\partial ph_n}}{{\partial d_l}}\times w_n 
+am_n\times ph_n\times\dfrac{{\partial w_n}}{{\partial d_l}},
\vspace{-0.25cm}
\end{equation}
where
\vspace{-0.05cm}
\begin{equation}\label{damd}
\dfrac{{\partial am_n}}{{\partial d_l}}= \dfrac{\pi\sqrt{\pi}\cos^2\phi_l}{4\lambda d_l^2 \alpha \sqrt{\alpha}},
\vspace{-0.25cm}
\end{equation}
\vspace{-0.25cm}
\begin{equation}\label{dphd}
\dfrac{{\partial ph_n}}{{\partial d_l}}=e^{ - j\left( {\frac{{{{ {\pi^2 \beta^2} }}}}{{\alpha}} - \frac{5\pi }{4}} \right)}\dfrac{{{e^{j\frac{\pi }{2}}}{\pi ^3}{\beta ^2}{{\cos }^2}\phi_l }}{{\lambda {d_l^2}{\alpha ^2} }},
\vspace{-0.25cm}
\end{equation}
and
\vspace{-0.25cm}
\begin{equation}\label{dwd}
\dfrac{{\partial w_n}}{{\partial d_l}}=\dfrac{{\sqrt \pi {e^{j\frac{{3\pi }}{4}}} {{\cos }^2}\phi_l }}{{\lambda {d_l^2}\alpha   }}\left( {\zeta_1{e^{j\zeta^2_2 }} + \zeta_2{e^{j\zeta^2_1}}} \right),
\vspace{-0.4cm}
\end{equation}
where $\zeta_1=\frac{{\alpha{D_y} + 2\pi \beta}}{{2\sqrt \alpha }}$ and $\zeta_2=\frac{{\alpha{D_y} - 2\pi \beta}}{{2\sqrt \alpha }}$.
Similarly, $\dfrac{{\partial a_n(d_l,\phi_l)}}{{\partial \phi_l}}$ can be obtained with
\vspace{-0.25cm}
\begin{equation}\label{damp}
\dfrac{{\partial am_n}}{{\partial \phi_l}}= -\dfrac{
	\pi \sqrt\pi \sin(2\phi_l)}{4\lambda d_l \alpha\sqrt\alpha},
\vspace{-0.25cm}
\end{equation}
\vspace{-0.4cm}
\begin{equation}\label{dphp}
\dfrac{{\partial ph_n}}{{\partial \phi_l}}= e^{ - j\left( {\frac{{{{ {\pi^2 \beta^2} }}}}{{\alpha}} - \frac{5\pi }{4}} \right)} \left(\dfrac{2\pi^2{e^{j\frac{\pi }{2}}}\beta\cos\phi}{\lambda \alpha}
+ \dfrac{\pi^3{e^{j\frac{\pi }{2}}}\beta^2\sin(2\phi)}{\lambda d \alpha^2}\right),
\vspace{-0.4cm}
\end{equation}
and
\vspace{-0.1cm}
\begin{equation}\label{dwp}
\dfrac{{\partial w_n}}{{\partial \phi_l}}=\dfrac{{\sqrt \pi {e^{j\frac{{3\pi }}{4}}} {{\sin }}(2\phi_l) }}{{\lambda {d_l}\alpha   }}\left( {\zeta_1{e^{j\zeta^2_2 }} + \zeta_2{e^{j\zeta^2_1}}} \right)
+ \dfrac{{2\sqrt \pi {e^{j\frac{{3\pi }}{4}}} {{\cos }}\phi_l }}{{\lambda \sqrt \alpha   }} \left( {{e^{j\zeta_2^2 }} - {e^{j\zeta_1^2 }}} \right).
\vspace{-0.25cm}
\end{equation}
\end{appendices}

\end{document}